\documentclass[12pt]{article}
\usepackage{natbib}
\usepackage{amssymb}
\usepackage{graphicx}
\usepackage{latexsym}

\newcommand{\blind}{1}

\usepackage[margin=1in]{geometry}

\usepackage{titletoc}

\usepackage{amsmath}
\usepackage{amsthm,thmtools}
\renewcommand\thmcontinues[1]{Continued}
\usepackage{amssymb} 
\usepackage{booktabs}
\usepackage{color}
\usepackage{tikz}
\usepackage{rotating}
\usepackage{enumitem}

\usepackage{algorithm}
\usepackage[noend]{algpseudocode}

\newtheorem{theorem}{Theorem}

\newtheorem{lemma}[theorem]{Lemma}

\theoremstyle{definition}

\newcommand{\openr}{\hbox{${\rm I\kern-.2em R}$}}
\newcommand{\openn}{\hbox{${\rm I\kern-.2em N}$}}

\newcommand{\Rem}{\operatorname{Rem}}

\newcommand{\Fnorm}[1]{\left\lVert#1\right\rVert_{\mathcal{F}_n}}

\newcommand{\LBn}{\textnormal{LB}_n}
\newcommand{\UBn}{\textnormal{UB}_n}

\DeclareMathOperator{\Var}{Var}
\DeclareMathOperator{\Corr}{Corr}
\DeclareMathOperator{\Cov}{Cov}
\DeclareMathOperator{\Prob}{Pr}

\DeclareMathOperator{\sgn}{sgn}
\newcommand{\Knstar}{\mathcal{K}_n^\star}

\allowdisplaybreaks

\usepackage[normalem]{ulem}
\usepackage{longtable}
\usepackage{authblk}

\begin{document}

\def\spacingset#1{\renewcommand{\baselinestretch}%
{#1}\small\normalsize} \spacingset{1}


\if1\blind
{
  \title{\bf Parametric-Rate Inference for One-Sided Differentiable Parameters}
  \author{Alexander R. Luedtke\thanks{
    This author gratefully acknowledges the support of the Berkeley Fellowship.}\\
    and \vspace{1em}\\
    Mark J. van der Laan\thanks{
    This author gratefully acknowledges the support of NIH grant R01 AI074345-06.}\vspace{1em}\\
    Division of Biostatistics, University of California, Berkeley}
  \maketitle
} \fi

\if0\blind
{
  \bigskip
  \bigskip
  \bigskip
  \begin{center}
    {\LARGE\bf Title}
\end{center}
  \medskip
} \fi

\bigskip
\begin{abstract}
Suppose one has a collection of parameters indexed by a (possibly infinite dimensional) set. Given data generated from some distribution, the objective is to estimate the maximal parameter in this collection evaluated at this distribution. This estimation problem is typically non-regular when the maximizing parameter is non-unique, and as a result standard asymptotic techniques generally fail in this case. We present a technique for developing parametric-rate confidence intervals for the quantity of interest in these non-regular settings. We show that our estimator is asymptotically efficient when the maximizing parameter is unique so that regular estimation is possible.
We apply our technique to a recent example from the literature in which one wishes to report the maximal absolute correlation between a prespecified outcome and one of $p$ predictors. The simplicity of our technique enables an analysis of the previously open case where $p$ grows with sample size. Specifically, we only require that $\log p$ grows slower than $\sqrt{n}$, where $n$ is the sample size. We show that, unlike earlier approaches, our method scales to massive data sets: the point estimate and confidence intervals can be constructed in $O(np)$ time.
\end{abstract}

\noindent%
\vfill

\newpage
\spacingset{1.45} 
\section{Introduction}
Many semiparametric and nonparametric estimation problems yield estimators which achieve a parametric rate of convergence. These estimators are often asymptotically linear, in that they can be written as an empirical mean of an influence function applied to the data. Valid choices of the influence function can be derived as gradients for a functional derivative of the parameter of interest. Applying the central limit theorem then immediately yields Wald-type confidence intervals which achieve the desired parametric rate. Such problems have been studied in depth over the past several decades \citep{Pfanzagl1990,Vandervaart91,Bickel1993,vdL02}.

While remarkably general, these approaches rely on the key condition that the parameter of interest is sufficiently differentiable for such a gradient to exist. Statisticians are increasingly encountering problems for which parametric-rate estimation is theoretically possible but the parameter is insufficiently differentiable to yield a standard first-order expansion demanded by older techniques. For example, suppose we observe baseline covariates, a binary treatment, and an outcome occuring after treatment. We wish to learn the mean outcome under the optimal individualized treatment strategy, i.e. the treatment strategy which makes treatment decisions which are allowed to use baseline covariate information to make treatment decisions \cite{Chakraborty&Moodie2013}. As another example, suppose we observe a vector of covariates $(X_1,\ldots,X_p)$ and an outcome $Y$. We wish give a confidence interval the maximal absolute correlation between a covariate $X_k$ and $Y$, or at least a lower bound on this quantity since this will suffice for a variable screening procedure. Alternatively, we may only with to test the null hypothesis that the maximal absolute correlation is zero. \cite{McKeague&Qian2015} provide a test of this null hypothesis using an adaptive resampling test (ART), a framework initially introduced in \cite{Laber&Murphy11} for estimating classification error.

These problems belong to a larger class of problems in which one observes $O_1,\ldots,O_n$ drawn independently from a $P_0$ in some (possibly nonparametric) statistical model $\mathcal{M}$ and wishes to estimate
\begin{align}
\Psi_n(P)&\equiv \max_{d\in\mathcal{D}_n}\Psi^d(P), \label{eq:estproblem}
\end{align}
at $P=P_0$, where $\mathcal{D}_n$ is an index set that may rely on sample size and each $\Psi^d$ is a sufficiently differentiable parameter to permit parametric-rate estimation using classical methods such as those presented in \cite{Bickel1993}. When there is no unique maximizer $d\in\mathcal{D}_n$ of $\Psi^d(P)$, then the inference problem is typically non-regular, in the sense that the parameter $P\mapsto \max_{d\in\mathcal{D}_n}\Psi^d(P)$ is not sufficiently differentiable to allow the use of standard techniques for obtaining root-$n$ inference. In these cases, the parameter of interest is termed one-sided pathwise differentiable \citep{Hirano&Porter2012}. In univariate calculus, functions such as $f(x) = \max \{x,0\}$ are one-sided differentiable at zero in that the left and right limits of $[f(x+\epsilon)-f(x)]/\epsilon$ are well-defined but disagree. The same holds for the $\Psi_n$ evaluated at a distribution $P_0$, but now the one-sided differentiability is caused by the subset of $\mathcal{D}_n$ containing the indices which maximize the expression on the right in (\ref{eq:estproblem}). A small fluctuation in $P_0$ can greatly reduce the subset of maximizing indices, leading to different derivatives depending on the fluctuation taken.

In this work, we present a method which, loosely, splits the sample in such a way that the estimated index in $\mathcal{D}_n$ which maximizes $\Psi^d(P_0)$ is conditioned on so that this estimated index need not have a limit. We do this iteratively to ensure that our estimator gets the full benefit of the sample size $n$. When the parameter is fixed with sample size and the $d$ maximizing $\Psi^d(P_0)$ is fixed, we show that our estimator is asymptotically efficient, and therefore also regular. Thus our estimator adapts to the non-regularity of the estimation problem.

Our estimator is inspired by the online estimator for pathwise differentiable parameters presented in \cite{vanderLaan&Lendle2014} and a subsequent modification of this estimator in \cite{Luedtke&vanderLaan2015b} to deal with the non-regularity when estimating the mean outcome under an optimal treatment rule. Such estimators are designed to be efficient in both computational complexity and storage requirements. We show that the estimator that we present in this work inherits many of these computational efficiency properties. We apply our technique to estimate the maximal absolute correlation considered in \cite{McKeague&Qian2015}. In this problem, we show that our estimator runs efficiently in both dimension and sample size, with a runtime of $O(np)$. In practice, this means that the lead author can implement our estimator using only \texttt{R} code and screen $p=100\,000$ variables using $n=1\,000$ samples on a single core of his laptop in under a minute. Thus our estimator seems to have both the statistically efficiency that has been demanded of estimators for generations and the computational efficiency that is becoming increasingly important in this new big data era.

\section{Estimator} \label{sec:probest}
We will now present our technique for a general estimation problem. Before doing so, we must introduce the notion of pathwise differentiability, since this provides the key object needed to construct our estimator.




\subsection{Pathwise differentiability}
We assume that each parameter $\Psi^d$, $d\in \mathcal{D}_n$ for any $n$, is pathwise differentiable for all distributions in our model \citep[see, e.g.,][]{Pfanzagl1990,Bickel1993}. For each $P\in\mathcal{M}$, we let $D^d(P)$ denote the canonical gradient of $\Psi^d$ at $P$. By definition $D^d(P)(O)$ is mean zero with finite variance under sampling from $P$. Typically pathwise differentiability implies that $\Psi^d$ satisfies the following linear expansion for any $P\in\mathcal{M}$ and $d\in\mathcal{D}_n$:
\begin{align}
\Psi^d(P) - \Psi^d(P_0)&= -\int D^d(P)(o) dP_0(o) + \Rem_n^d(P), \label{eq:psi1stord}
\end{align}
where we omit the dependence of $\Rem_n^d(P)$ on $P_0$ in the notation and indicate its possible dependence on sample size with the subscript $n$. Above $\Rem_n^d(P)$ is a second-order remainder term that is small whenever $P$ is close to $P_0$. We consider this condition more closely in our examples, but for non-sample size dependent parameters this term can typically be made to be $O_{P_0}(1/n)$ in a parametric model and often can be made to be $o_{P_0}(1/\sqrt{n})$ in a nonparametric model. For a more thorough presentation, see \cite{Pfanzagl1990} or \cite{Bickel1993}. 

\subsection{Estimator and confidence interval}
We now present a stabilized one-step estimator for problems of the type found in (\ref{eq:estproblem}) when the required differentiability condition on $\Psi^d$ holds. 

Let $\{\ell_n\}$ be some sequence such that $n-\ell_n\rightarrow \infty$. One possible choice is $\ell_n=0$ for all $n$. For each $j=\ell_n,\ldots,n-1$, let $d_{nj}$ represent an estimate of a maximizer of (\ref{eq:estproblem}), $\hat{P}_{nj}$ be an estimate of the likelihood obtained using observations $\left(O_i : i=1,\ldots,j\right)$, and $\hat{D}_{nj}$ equal $D^d(P)$ evaluated at $P=\hat{P}_{nj}$ and $d=d_{nj}$. For nonnegative weights $w_{\ell_n},\ldots,w_{n-1}$ that we will define shortly with $\sum_{j=\ell_n}^{n-1} w_j=n-\ell_n$, our stabilized one-step estimator takes the form
\begin{align*}
\psi_n&\equiv \frac{1}{n-\ell_n}\sum_{j=\ell_n}^{n-1} w_{j}\left[\Psi^{d_{nj}}(\hat{P}_{nj}) + \hat{D}_{nj}(O_{j+1})\right].
\end{align*}
Our proposed 95\% confidence interval has the form
\begin{align*}
\left[\LBn,\UBn\right]&\equiv\left[\psi_n\pm 1.96\frac{\bar{\sigma}_n}{\sqrt{n-\ell_n}}\right],
\end{align*}
where we will define $\bar{\sigma}_n$ momentarily and one can replace 1.96 by the desired quantile of the normal distribution to modify the confidence level.

We now define the weights. Let $\hat{\sigma}_{nj}^2$ represent an estimate of the variance of $\hat{D}_{nj}(O)$, $O\sim P_0$, conditional on observations $O_1,\ldots,O_j$. This estimate should only rely on those $j$ observations. Often we can let
\begin{align*}
\hat{\sigma}_{nj}^2&\equiv\frac{1}{j}\sum_{i=1}^{j} \left[\hat{D}_{nj}(O_i)-\frac{1}{j}\sum_{i=1}^j \hat{D}_{nj}(O_{i})\right]^2,\;j=\ell_n,\ldots,n-1.
\end{align*}
The standard deviation type variable in the confidence interval definition is given by $\bar{\sigma}_n\equiv \left(\frac{1}{n-\ell_n}\sum_{j=\ell_n}^{n-1} \hat{\sigma}_{nj}^{-1}\right)^{-1}$, and the weights are given by $w_j\equiv \bar{\sigma}_n\hat{\sigma}_{nj}^{-1}$, where we have omitted the possible dependence of the weights on sample size in the notation. 

Our estimator $\psi_n$ is similar to the online one-step estimator developed in \cite{vanderLaan&Lendle2014} for streaming data, but it weights each term proportionally to the estimated inverse standard deviation of $\hat{D}_{nj}(O)$ when $O\sim P_0$. Our confidence interval takes a form similar to a Wald-type confidence interval, but replaces the typical standard deviation with $\bar{\sigma}_n$ and has width on the order of $1/\sqrt{n-\ell_n}$ rather than $1/\sqrt{n}$. Note of course that $\ell_n=o(n)$ implies that $1/\sqrt{n-\ell_n}-1/\sqrt{n}$ converges to zero.

\subsection{First main result: validity of confidence interval}
We now prove the validity of our confidence interval. Let $\sigma_{nj}^2\equiv \Var_{P_0}\left(\hat{D}_{nj}(O)|O_1,\ldots,O_j\right)$. The validity of the lower bound of the confidence interval relies on the following conditions:
\begin{enumerate}[noitemsep,label=C\arabic*)]
  \item There exists some $M<\infty$ such that $\frac{1}{n-\ell_n}\sum_{j=\ell_n}^{n-1} P_0\left(\left.\frac{|\hat{D}_{nj}(O)|}{\hat{\sigma}_{nj}}<M\right|O_0,\ldots,O_{j-1}\right)\rightarrow 1$ in probability as $n\rightarrow\infty$. \label{it:C1}
  \item $\frac{1}{n-\ell_n}\sum_{j=\ell_n}^{n-1}\left|\frac{\sigma_{nj}^2}{\hat{\sigma}_{nj}^2}-1\right|\rightarrow 0$ in probability as $n\rightarrow\infty$. \label{it:C2}
  \item $\frac{1}{\sqrt{n-\ell_n}}\sum_{j=\ell_n}^{n-1}\hat{\sigma}_{nj}^{-1}\widehat{\Rem}_{nj}\rightarrow 0$ in probability as $n\rightarrow\infty$, where $\widehat{\Rem}_{nj}\equiv \Rem^{d_{nj}}(\hat{P}_{nj})$. \label{it:C3}
\end{enumerate}
The validity of the upper bound requires the following additional condition:
\begin{enumerate}[noitemsep,label=C\arabic*),start=4]
  \item $\frac{1}{\sqrt{n-\ell_n}}\sum_{j=\ell_n}^{n-1} \hat{\sigma}_{nj}^{-1}\left[\Psi^{d_{nj}}(P_0)-\Psi_n(P_0)\right]$ converges to zero in probability as $j\rightarrow\infty$. \label{it:C4}
\end{enumerate}
We now present our main result. We discuss the conditions in Section \ref{sec:discconds}.
\begin{theorem}[Validity of confidence interval] \label{thm:mainresult}
If \ref{it:C1}, \ref{it:C2}, and \ref{it:C3} hold, then
\begin{align*}
\liminf_{n\rightarrow\infty}\Prob\left(\Psi_n(P_0)\ge \LBn\right)&\ge 1-\alpha/2.
\end{align*}
If \ref{it:C4} also holds, then
\begin{align*}
\lim_{n\rightarrow\infty}\Prob\left(\LBn\le \Psi_n(P_0)\le \UBn\right)&= 1-\alpha.
\end{align*}
\end{theorem}
\begin{proof}
The definition $\psi_n$ combined with (\ref{eq:psi1stord}) yield that
\begin{align}
\sqrt{n-\ell_n}\bar{\sigma}_n^{-1}\left[\psi_n-\Psi_n(P_0)\right]=&\; \frac{1}{\sqrt{n-\ell_n}}\sum_{j=\ell_n}^{n-1} \hat{\sigma}_{nj}^{-1}\left(\hat{D}_{nj}(O_{j+1}) -E_{P_0}\left[\hat{D}_{nj}(O)|O_1,\ldots,O_{j}\right]\right) \label{eq:clt} \\
&+ \frac{1}{\sqrt{n-\ell_n}}\sum_{j=\ell_n}^{n-1} \hat{\sigma}_{nj}^{-1}\left[\Psi^{d_{nj}}(P_0)-\Psi_n(P_0) + \widehat{\Rem}_{nj}\right]. \nonumber
\end{align}
The second line converges to zero in probability by \ref{it:C3} and \ref{it:C4}. By \ref{it:C1}, \ref{it:C2}, and the martingale central limit theorem for triangular arrays in \cite{Gaenssleretal1978}, (\ref{eq:clt}) converges in distribution to a standard normal random variable. A standard Wald-type confidence interval construction argument shows that the confidence interval has coverage approaching $1-\alpha$ under \ref{it:C1} through \ref{it:C4}.

Now suppose \ref{it:C4} does not hold. By (\ref{eq:estproblem}), $\sum_{j=\ell_n}^{n-1} \hat{\sigma}_{nj}^{-1}\left[\Psi^{d_{nj}}(P_0)-\Psi_n(P_0)\right]\le 0$. The same argument readily shows the validity of the lower bound under only \ref{it:C1}, \ref{it:C2}, and \ref{it:C3}.
\end{proof}

\subsection{Second main result: efficiency when the maximizer in (\ref{eq:estproblem}) is unique}
We have presented a parametric-rate estimator for $\Psi_n(P_0)$, but thus far we have not made any claims about the efficiency of our estimator. In this section, we consider a fixed parameter in (\ref{eq:estproblem}) that does not rely on sample size. We therefore omit the $n$ subscript in many quantities to indicate their lack of dependence on sample size. We will give conditions under which our estimator is asymptotically efficient among all regular, asymptotically linear estimators. The efficiency bound is not typically well-defined when the maximizer is non-unique due to the non-regularity of the problem -- generally in this case no regular, asymptotically linear estimator exists, so neither does an efficient member of this class \citep{Hirano&Porter2012}. Thus the conditions that we give in this section will typically only hold when the maximizer $d_0\in\mathcal{D}$ in (\ref{eq:estproblem}) is unique.

We use the following additional assumptions for our efficiency result:
\begin{enumerate}[noitemsep,label=C\arabic*),start=5]
  \item $E_{P_0}\left[\left.\left(\hat{D}_j(O)-D^{d_0}(P_0)(O)\right)^2\right|O_1,\ldots,O_j\right]\rightarrow 0$ in probability as $j\rightarrow\infty$. \label{it:C5}
  \item There exists some $M<\infty$ such that $P_0\left(D^{d_0}(P_0)(O)<M\right)$ and $P_0\left(\hat{D}_j(O)<M\right)$ with probability approaching $1$ as $j\rightarrow\infty$. \label{it:C6}
  \item $\inf_{j\ge 1}\hat{\sigma}_j^2>\gamma$ with probability $1$ over draws of $(O_j : j=0,1,\ldots)$. \label{it:C7}
\end{enumerate}
We discuss the conditions immediately following the theorem.
\begin{theorem}[Asymptotic efficiency]
Suppose that $\Psi$ does not depend on sample size and is pathwise differentiable with canonical gradient $D^{d_0}(P_0)$. Further suppose that $\ell_n=o(n)$. If \ref{it:C1} through \ref{it:C7} hold, then
\begin{align*}
\bar{\sigma}_n^2\rightarrow \Var_{P_0}\left(D^{d_0}(P_0)(O)\right)\textnormal{ in probability as }n\rightarrow\infty.
\end{align*}
Furthermore,
\begin{align*}
\psi_n-\Psi(P_0) = \frac{1}{n}\sum_{i=1}^n D^{d_0}(P_0)(O_i) + o_{P_0}(n^{-1/2}).
\end{align*}
Thus $\psi_n$ is asymptotically efficient among all regular, asymptotically linear estimators.
\end{theorem}
The proof is entirely analogous to the proof of Corollary 3 in \cite{Luedtke&vanderLaan2015b} so is omitted.

The additional conditions needed for this result over Theorem \ref{thm:mainresult} are mild when the maximizing index is unique. Condition \ref{it:C5} says that $\Psi$ should have the same canonical gradient as $\Psi^{d_0}$. While this should be manually checked in each example, it will be fairly typical when the maximizer is unique, since in this case an arbitrarily small fluctuation of $P_0$ will generally not change the maximizer. This is similar to problems in introductory calculus where the derivative at the maximum is zero. Condition \ref{it:C5} requires that $\hat{D}^{j}(O)$ converge to $D^{d_0}(P_0)(O)$ in mean-squared error, which is to be expected if $\hat{P}_{nj}$ begins to approximate $P_0$ and $d_{nj}$ converges to the unique maximizer $d_0$ as $n,j\rightarrow\infty$. Condition \ref{it:C6} is a bounding assumption on the canonical gradient and estimates thereof that will hold in many examples of interest. Finally, Condition \ref{it:C7} will hold if one knows that $\Var_{P_0}\left[D^d(P)(O)\right]$ is bounded away from zero uniformly in $P\in\mathcal{M}$ and $d\in\mathcal{D}$, and uses this knowledge to truncate $\hat{\sigma}_j^2$ at $\gamma_j>0$ for some deterministic sequence $\gamma_j\rightarrow 0$. For $\gamma_j$ sufficiently small and $j$ sufficiently large this truncation scheme will then have no effect on the variance estimates $\hat{\sigma}_j^2$. 

\subsection{Discussion of conditions of Theorem \ref{thm:mainresult}} \label{sec:discconds}
In this section, we again consider the setting where the parameter does not depend on sample size, and consequently omit the $n$ subscript to quantities which no longer depend on sample size. We will show that \ref{it:C7} and the following conditions imply the conditions of Theorem \ref{thm:mainresult}:
\begin{enumerate}[noitemsep,label=C\arabic*),start=9]
  \item $\hat{\sigma}_j^2 - \sigma_{j}^2$ converges to zero in probability as $j\rightarrow\infty$. \label{it:C9}
  \item $\sqrt{j}\widehat{\Rem}_{j}\equiv \sqrt{j}\Rem^{d_{j}}(\hat{P}_{j})$ converges to zero in probability as $j\rightarrow\infty$. \label{it:C10}
\end{enumerate}
The validity of the upper bound requires the following additional condition:
\begin{enumerate}[noitemsep,label=C\arabic*),start=11]
  \item $\sqrt{j}\left[\Psi^{d_{j}}(P_0)-\Psi(P_0)\right]$ converges to zero in probability as $j\rightarrow\infty$. \label{it:C11}
\end{enumerate}
For simplicity, we will take $\ell_n=0$ in this section.

We now discuss the conditions. Condition \ref{it:C1} is an immediate consequence of \ref{it:C7} and $D^d(P)(o)$ being uniformly bounded in $P\in\mathcal{M}$, $d\in\mathcal{D}$, $o\in\mathcal{O}$. This will be plausible in many situations, including the examples in this paper. A more general Lindeberg-type condition also suffices \citep[see Condition C1 in][]{Luedtke&vanderLaan2015b}, though we omit its presentation here for brevity.

The other three conditions all rely on terms like $\frac{1}{n}\sum_{j=0}^{n-1} R_j$ converging to zero in probability, possibly at some rate. Ideally we want a stochastic version of the fact that, for $\beta\in[0,1)$,
\begin{align}
\frac{1}{n}\sum_{j=1}^n j^{-\beta}&\approx \frac{1}{n}\int_1^n j^{-\beta}dj\approx \frac{n^{-\beta}}{1-\beta}\,\mbox{ when $n$ is large.} \label{eq:sumorder}
\end{align}
Lemma 6 of \cite{Luedtke&vanderLaan2015b} establishes this result. We restate it here for convenience.
\begin{lemma}[Lemma 6 in \citealp{Luedtke&vanderLaan2015b}] \label{lem:seriesconv}
Suppose that $R_j$ is some sequence of (finite) real-valued random variables such that $R_j=o_{P_0}(j^{-\beta})$ for some $\beta\in[0,1)$, where we assume that each $R_j$ is a function of $\{O_i : 1\le i\le j\}$. Then,
\begin{align*}
\frac{1}{n}\sum_{j=0}^{n-1} R_j&= o_{P_0}\left(n^{-\beta}\right).
\end{align*}
\end{lemma}
Conditions \ref{it:C2} through \ref{it:C4} are now easily handled. Condition \ref{it:C2} is a consequence of the fact that
\begin{align*}
\frac{1}{n}\sum_{j=0}^{n-1}\left|\frac{\sigma_{j}^2}{\hat{\sigma}_j^2}-1\right|&\le \gamma^{-1}\frac{1}{n}\sum_{j=0}^{n-1}\left|\hat{\sigma}_j^2-\sigma_{j}^2\right|\rightarrow 0\textnormal{ in probability as }n\rightarrow\infty,
\end{align*}
where the inequality holds by \ref{it:C7} and the convergence holds by \ref{it:C9} Lemma \ref{lem:seriesconv}. Condition \ref{it:C9} is easily shown to hold under Glivenko-Cantelli conditions on the estimators $\hat{P}_j$ and $d_j$ \citep[see, e.g., Theorem 7 in][]{Luedtke&vanderLaan2015b}. Conditions \ref{it:C3} and \ref{it:C4} are an immediate consequence of \ref{it:C10} and \ref{it:C11} combined with Lemma \ref{lem:seriesconv}.

While sufficient conditions for \ref{it:C11} should be developed in each individual example, we can give intuition as to why this condition should be reasonable.
For any $P\in\mathcal{M}$, let $d(P)$ return a maximizer of (\ref{eq:estproblem}). We are interested in ensuring that $\Psi^{d_n}(P_0)-\Psi^{d(P_0)}(P_0)$ is small, where $d_n$ is our estimate of a maximizer of (\ref{eq:estproblem}). This can be expected to hold when the parameter $P\mapsto \Psi^{d(P)}(P_0)$ has pathwise derivative zero at $P=P_0$, where the $P_0$ in the $\Psi$ argument is fixed. When well-defined, the pathwise derivative will be zero because $d(P)$ is chosen to maximize $\Psi^d(P_0)$ in $d$.

\subsection{Computationally efficient implementation}
There are several computationally efficient ways to compute our estimator. In Section 6.1 of \cite{Luedtke&vanderLaan2015b}, we show that the runtime of our estimator can be dramatically improved by running the algorithm used to compute each $\hat{P}_j$ a limited number of times, say ten times. We do not detail this approach here, though we note that the theorems we have presented are general enough to apply to this case.

An alternative approach to improve runtime is to use the estimator's online nature to compute it efficiently both in time and storage. Suppose that we have an algorithm to update the estimate $\hat{P}_{nj}$ of $P_0$ to the estimate $\hat{P}_{n(j+1)}$ based on the first $j$ observations by looking at $O_{j+1}$ only. This will often be feasible if the parameter of interest and the bias correction step only require estimates of certain components of $P_0$, e.g. of a set of regression and classification functions. In these cases we can apply modern regression and classification approaches to estimate these quantities \citep[see, e.g.,][]{Xu2011,Lutsetal2013}. Often $d_{nj}$ can also be obtained using online methods, and thus $\frac{1}{n-\ell_n}\sum_{j=\ell_n}^{n-1}\left[\Psi^{d_{nj}}(\hat{P}_{nj})+\hat{D}_{nj}(O_{j+1})\right]$ can be estimated online by keeping a running sum. This quantity is not equal to $\psi_n$ because it does not yet include the weights.

It will not in general be possible to compute the weights online, though their computation does not require storing $O(n)$ observations in memory. We can estimate $\Var_{P_0}(\hat{D}_{nj}(O))$ consistently using the $r_j$ observations, where $r_j\rightarrow\infty$ but can grow very slowly (even $\log j$ suffices asymptotically, though such a slow growth is not recommended for finite samples). Given online estimates of these variances, it is then straightforward to compute both $\bar{\sigma}_n$ and the weights and incorporate these into our estimator. In some cases, we can compute the weights, and thus the estimator, in a truly online fashion. Describing general sufficient conditions for this appears to be difficult, but we conjecture that often this will not typically hold if $\mathcal{D}_n$ is not of finite cardinality. The weights can be computed online in the maximal correlation example that we describe in the next section.

\section{Maximal correlation example}
\subsection{Problem formulation}
We now present the running example of this work, namely the maximal correlation estimation problem considered by \cite{McKeague&Qian2015}. The observed data structure is $O=(X,Y)$, where $X=(X_k : k=1,\ldots)$ is a $[-1,1]^\infty$ vector of predictors and $Y$ is an outcome in $[-1,1]$. For each $n$, we let $\mathcal{K}_n$ represent a subset of these predictors of size $p$, where throughout we assume that
\begin{align}
\beta_n^2&\equiv \frac{\log p}{\sqrt{n}}\rightarrow 0\textnormal{ as }n\rightarrow\infty. \label{eq:pnrate}
\end{align}
For readability, we omit the dependence of $p$ on $n$ in the notation. Under a distribution $P$, the maximal absolute correlation of a predictor with $Y$ is given by
\begin{align}
\Psi_n(P)&\equiv \max_{k\in\mathcal{K}_n}\left|\Corr_P(X_k,Y)\right|, \label{eq:psikdef}
\end{align}
where $\Corr_P(X_k,Y)$ is the correlation of $X_k$ and $Y$ under $P$. We wish to develop confidence intervals for $\Psi_n(P_0)$. When a test of $H_0 : \Psi_n(P_0)=0$ against the complementary alternative, we also wish to establish the behavior of our test against local alternatives as was done in \cite{McKeague&Qian2015}.

In contrast to \cite{McKeague&Qian2015}, the procedure that we present in this work:
\begin{enumerate}[noitemsep,label=\arabic*)]
  \item is proven to work when $p$ grows with sample size at any rate satisfying (\ref{eq:pnrate});
  \item yields confidence intervals for the maximal correlation rather than just a test of the null hypothesis that it is equal to zero,
  \item allows a non-null the maximizer in (\ref{eq:psikdef}) to be non-unique; \label{it:uniquemax}  
  \item is proven to work in a nonparametric model that neither assumes linearity nor homoscedasticity.
\end{enumerate}
While \citeauthor{McKeague&Qian2015} argued that \ref{it:uniquemax} is unlikely in practice, having two non-null maximizers be approximately equal may still have finite sample implications for their test in some settings. 

We now show that this problem fits in our framework. To satisfy the pathwise differentiability condition, we let $\mathcal{D}_n\equiv\mathcal{K}_n\times \{-1,1\}$ and, for each $d=(k,m)\in\mathcal{D}_n$,
\begin{align*}
\Psi^d(P)&\equiv m \Corr_P(X_k,Y).
\end{align*}
Note that $\Psi_n(P)$ now takes the form in (\ref{eq:estproblem}), where we note that the use of $m$ in the definition of $\Psi^d$ serves to ensure that $\Psi_n(P_0)$ represents the correlation with the maximal \textit{absolute value}.

\subsection{Differentiability condition} \label{sec:corrdiff}
\subsubsection*{Canonical gradients}
For each $k$, let $s_P^2(X_k)\equiv \Var_P(X_k)$, and likewise for $s_P^2(Y)$. For ease of notation we let $s_0^2(X_k)\equiv s_{P_0}^2(X_k)$, and likewise for $s_0^2(Y)$ and $\Corr_0(X_k,Y)$. An application of the delta method shows that $\Psi^d$ has canonical gradient $D^d(P)(o)$ given by
\begin{align*}
m\times\left(\frac{\left(x_k-E_P[X_k]\right)\left(y-E_P[Y]\right)}{s_P(X_k) s_P(Y)} - \frac{1}{2}\Corr_P(X_k,Y)\left[\frac{\left(x_k-E_P[X_k]\right)^2}{s_P^2(X_k)} + \frac{\left(y-E_P[Y]\right)^2}{s_P^2(Y)}\right]\right).
\end{align*}
In order to ensure that $D^d(P_0)$ is uniformly bounded for all $d$, we assume throughout that, for some $\delta\in(0,1]$,
\begin{align*}
\min\{s_0(Y),s_0(X_1),s_0(X_2),\ldots\}>\delta.
\end{align*}
\subsubsection*{Second-order remainder}
Fix $d=(k,m)\in\mathcal{D}_n$ and $P\in\mathcal{M}$. Let $\tilde{\delta}\in(0,1]$ be some constant such that both $s_P(X_k)$ and $s_P(Y)$ are larger than $\tilde{\delta}$. Lemma \ref{lem:rembd} in the appendix proves that
\begin{align}
\left|\Rem^d(P)\right|\le \tilde{\delta}^{-1}\Bigg(&\left|s_P(X_k)s_P(Y)-s_0(X_k)s_0(Y)\right|\left|\Corr_P(X_k,Y)-\Corr_0(X_k,Y)\right| \nonumber \\
&+ \left(E_P[X_k]-E_{P_0}[X_k]\right)^2 + \left(E_P[Y]-E_{P_0}[Y]\right)^2 \nonumber \\
&+ \frac{s_0^2(Y)}{s_P^2(Y)}\left[s_P(X_k)-s_0(X_k)\right]^2 + \frac{s_0^2(X_k)}{s_P^2(X_k)}\left[s_P(Y)-s_0(Y)\right]^2\Bigg). \label{eq:rembd}
\end{align}
The first term above is small if $s_P(X_k)$, $s_P(Y)$, and $\Corr_P(X_k,Y)$ are close to $s_0(X_k)$, $s_0(Y)$, and $\Corr_0(X_k,Y)$. The middle terms are small if $E_P[X_k]$ and $E_P[Y]$ are close to $E_{P_0}[X_k]$ and $E_{P_0}[Y]$. The final terms are small if $s_P(X_k)$ and $s_P(Y)$ are close to $s_0(X_k)$ and $s_0(Y)$. 
\subsubsection*{Variance of canonical gradients}
There is no elegant (and informative) expression for the variance $\sigma_{nj}^2$ of $D^{d_{nj}}(P_j)(O)$. Nonetheless, we show in Lemma \ref{lem:spbds} of the appendix that our estimates $\hat{\sigma}_{nj}^2$, taken as the sample variance of $D^{d_{nj}}(P_j)(O)$, concentrate tightly about $\sigma_{nj}^2$ with high probability when the sample size is large enough. Thus, in practice, one can actually check if $\sigma_{nj}^2$ is small by looking at $\hat{\sigma}_{nj}^2$.
If $P_0$ is normal, then this variance is equal to $\left[1-\Corr_{P_0}(X_k,Y)^2\right]^2$, and so is only zero if $\Corr_{P_0}(X_k,Y)=1$. Though such an elegant expression does not exist for the variance of $D^d(P_0)(O)$ for general distributions, one can still show in general that the variance of $D^d(P_0)$ is equal to zero only if $\Corr_{P_0}(X_k,Y)=1$.
Here we make the slightly stronger assumption that
\begin{align}
\inf_{n\ge 2} \min_{(k,m)\in\mathcal{K}_n\times\{-1,1\}} \Var_{P_0}\left(D^d(P_0)(O)\right) \ge\gamma>0. \label{eq:gamma}
\end{align}

\subsection{Our estimator}
We will use the estimator presented in Section \ref{sec:probest} to estimate $\Psi_n(P_0)$. At each index $j\ge\ell_n$ we use the empirical distribution $P_j$ of the observations $O_1,\ldots,O_j$ to estimate $P_0$. We estimate $\sigma_{nj}^2$ with the variance of $\hat{D}_{nj}(O)$ under $P_j$.

In the appendix, we detail conditions on $\ell_n$ which ensure that $\ell_n$ does not grow too slowly or quickly. For any $\epsilon\in(0,2)$, one possible choice of $\ell_n$ that satisfies these conditions is
\begin{align}
\ell_n&= \max\left\{(\log\max\{n,p\})^{1+\epsilon},n\exp(-\beta_n^{-2+\epsilon})\right\}. \label{eq:ellnfastenough}
\end{align}
We show that this choice of $\ell_n$ ensures \ref{it:C1}, \ref{it:C2}, and \ref{it:C3} in the appendix. By Theorem \ref{thm:mainresult} this establishes the validity of the lower bound of our confidence interval. 
We can also show that this lower bound is tight up to a term of the order $n^{-1/4}\beta_n$.
\begin{theorem}[Tightness of the lower bound] \label{thm:mcklb}
For any sequence $t_n\rightarrow\infty$, $\Psi_n(P_0)< \LBn + t_nn^{-1/4}\beta_n$ with probability approaching $1$.
\end{theorem}
We now consider the validity of the upper bound of our confidence interval, which holds under \ref{it:C4}. This condition is trivially valid if $\Psi_n(P_0)=0$ for all $n$. Condition \ref{it:C4} is also valid under the following margin condition:
\begin{enumerate}[label=MC)]
  \item \label{it:marginmain} For some sequence $t_n\rightarrow\infty$, there exists a sequence of non-empty subsets $\Knstar\subseteq\mathcal{K}_n$ such that, for all $n$,
\begin{align*}
&\sup_{k\in\Knstar} \left|\Corr_{P_0}(X_k,Y)\right| - \inf_{k\in\Knstar}\left|\Corr_{P_0}(X_k,Y)\right|= o(n^{-1/2}), \\
&\inf_{k\in \Knstar}\left|\Corr_{P_0}(X_{k},Y)\right|\ge \sup_{k\in \mathcal{K}_n\backslash \Knstar}\left|\Corr_{P_0}(X_{k},Y)\right| + t_n n^{-1/4}\beta_n.
\end{align*}
If $\Knstar=\mathcal{K}_n$, then the supremum over $\mathcal{K}_n\backslash\Knstar$ is taken to be zero.
\end{enumerate}
\begin{theorem}[Validity of the upper bound] \label{thm:mckub}
If \ref{it:marginmain} or $\Psi_n(P_0)=0$ for all $n$, then \ref{it:C4} holds so that $\LBn\le \Psi_n(P_0)\le \UBn$ with probability approaching $1-\alpha$.
\end{theorem}
We outline the techniques used to prove these two results at the end of this subsection. Complete proofs are given in the appendix.

Suppose we wish to test $H_0 : \Psi_n(P_0)=0$ against $H_1 : \Psi_n(P_0)>0$. Consider the test that rejects $H_0$ if $\LBn>0$. We wish to explore the behavior of this test under local alternatives where $\Psi_n(P_0)$ converges to zero slower than $n^{-1/4}\beta_n$. Theorem \ref{thm:mcklb} shows that this test has power converging to one under such local alternatives. Furthermore, as the lower bound is valid in general, this test has type I error of at most $\alpha/2$ under the null. This is indeed an exciting result as it enables the study of local alternatives even when dimension grows quickly with sample size. If dimension does not grow with sample size, this shows that we can detect against any alternatives converging to zero slower than $n^{-1/2}\sqrt{\log n}$. We would not be surprised if the $\sqrt{\log n}$ is unnecessary, but rather that it is simply a result of our proof techniques which give high probability bounds on the concentration of our correlation estimates at each sample size. \cite{McKeague&Qian2015} showed that their method is consistent against a class of alternatives converging to zero slower than $n^{-1/2}$ provided the optimal index is unique. Our result does not rely on this uniqueness condition.

Theorem \ref{thm:mckub} shows that the upper bound of our confidence interval is also valid under a reasonable margin condition. The margin condition states that there may be many non-null approximate maximizers provided their absolute correlations are well-separated from the absolute correlations of the other predictors with $Y$. By ``approximate'' we mean that their absolute correlations all fall within $o(n^{-1/2})$ of one another. If $\mathcal{K}_n$ does not depend on sample size, then this theorem shows that our two-sided confidence interval is always valid.

\begin{proof}[Sketch of proofs of Theorems \ref{thm:mcklb} and \ref{thm:mckub}]
Our proofs of both of these theorems rely on high-probability bounds of the absolute differences between our estimates of $s_{P_j}^2(X_k)$, $s_{P_j}^2(Y)$, $\Corr_{P_j}(X_k,Y)$, and $\hat{\sigma}_{nj}$ and their population counterparts, uniformly over $k\in\mathcal{K}_n$ and $j$. We show that, with probability at most $1-1/n$, all of these absolute differences are upper bounded by constants (with explicit dependence on $\gamma$ and $\delta$) times $j^{-1/2}\log\max\{n,p\}$.

Condition \ref{it:C1} follows once we show that, with high probability, $s_{P_j}^2(X_k)$ and $s_{P_j}^2(Y)$ are bounded below by $\delta/2$ and $\hat{\sigma}_{nj}^2$ is bounded below by $\gamma/2$ uniformly over $j\ge\ell_n$ for $n$ large enough. Condition \ref{it:C2} and \ref{it:C3} are easy consequences of our concentration results. The concentration results also yield that
\begin{align*}
\frac{1}{n-\ell_n}\sum_{j=\ell_n}^{n-1} \hat{\sigma}_{nj}^{-1}\left[\Psi^{d_{nj}}(P_0)-\Psi_n(P_0)\right]&= O_{P_0}\left(n^{-1/4}\beta_n\right),
\end{align*}
which then quickly yields Theorem \ref{thm:mcklb} thanks to the expression in (\ref{eq:clt}).

Now suppose \ref{it:marginmain} holds. By our concentration inequalities, we select a $k_{nj}\in\Knstar$ for each $j\ge Ct_n^{-1} n$ with high probability, where $C$ is a constant. We also correctly specify $m_{nj}$ to be the sign of $\Corr_0(X_{k_{nj}},Y)$. Because all of the absolute correlations in $\Knstar$ are small, the difference between $\Psi^{d_{nj}}(P_0)$ for $d_{nj}=(k_{nj},m_{nj})$ and $\Psi_n(P_0)$ is very small. If $\ell_n< Ct_n^{-1}$, then we can apply our concentration inequalities to establish that these first few values of $j$ for which $j<Ct_n^{-1}$ are small enough so that \ref{it:C4} still holds, yielding Theorem \ref{thm:mckub}.
\end{proof}
In Appendix \ref{app:compcomplex}, we show that our estimator runs in $O(np)$ time. We show that the estimator can be computed using $O(p)$ storage when the observations $O_1,\ldots,O_n$ arrive in a data stream. This result is closely related to the fact that, for a $\mathbb{R}^p$-valued sequence $\{t_i\}$, the sum $S_j\equiv \sum_{i=1}^j t_i$ at $j=n$ can be computed in time $O(np)$ using storage $O(p)$. In particular, one can use the recursion relation $S_j=t_j + S_{j-1}$, thereby only storing $t_j$ and $S_{j-1}$ when computing $S_j$. Our estimator can also be computed in $O(np)$ time and $O(n)$ storage when the vectors $(X_{jr} : j=1,\ldots,n)\in\mathbb{R}^n$ arrive in a stream for $r=1,2,\ldots,p$, where $X_{jr}$ is the observation of $X_r$ for individual $j$. We do not prove the $O(n)$ storage result in the appendix due to space constraints, though the algorithm is closely related to that given in Appendix \ref{app:compcomplex}.

\section{Simulation study}
We now consider the power and scalability of our method using the simulations similar to those described in \cite{McKeague&Qian2015}. Let $X\sim \textnormal{MVN}(0,\Sigma)$ for $\Sigma$ a $p\times p$ covariance matrix to be given shortly, and $\tau_1,\ldots,\tau_p$ be a sequence of i.i.d. nromal random variables independent of all other quantities under consideration. We will use two types of errors: the homoscedastic error $\tau_1$ and the heteroscedastic error $\eta(X)\equiv \sum_{k=1}^p X_k \tau_k/\sqrt{p}$. For $(n,p)=(200,200),(500,2\,000)$, we generate data using the following distributions: (N.IE) $Y=\tau_1$, (A1.IE) $Y=X_1/5 + \tau_1$, (A2.IE) $Y=0.15\sum_{k=1}^5 X_k -0.1\sum_{k=6}^{10} X_k + \tau_1$, (N.DE) $Y=\eta(X)$, (A1.DE) $Y=X_1/5 + \eta(X)$, and (A2.DE) $Y=0.15\sum_{k=1}^5 X_k -0.1\sum_{k=6}^{10} X_k + \eta(X)$. For $(n,p)=(2\,000,30\,000)$, we generate data using the following distributions: (N.IE) $Y=\tau_1$, (A3.IE) $Y = X_1/15 + \tau_1$, and (A4.IE) $Y =0.03\sum_{k=1}^5 X_k -0.015\sum_{k=6}^{10} X_k + \tau_1$. We set all of the diagonal elements in the covariance matrix $\Sigma$ equal $1$, and the off-diagonal elements equal $\rho$, where for each simulation setting we let $\rho=0,0.25,0.5,0.75$.

We conduct a 5\% test of $\Psi(P_0)>0$ by checking if the lower bound of a 90\% confidence interval for this quantity is greater than zero. We use models N.IE and N.DE to evaluate type I error and all other models evaluate power. We run our method with $\ell_n$ as in (\ref{eq:ellnfastenough}), where we let $\epsilon=0.5$. For ease of implementation, we compute our method on chunks of data of size $(n-\ell_n)/10$ (see Section 6.1 of \citealp{Luedtke&vanderLaan2015b}). We compare our method to the parametric bootstrap analogue of ART described in Section 2 of \cite{Zhang&Laber2015} for all $n$. The parametric bootstrap analogue of ART assumes a locally linear model with homoscedastic errors. We use $500$ bootstrap draws for each run of the parametric bootstrap procedure. \citeauthor{Zhang&Laber2015} show that their method, which does not involve running a computationally burdensome double bootstrap procedure, has comparable performance to ART across sample sizes and predictor dimension, while being more computationally efficient. For this reason, we do not directly compare against the ART results in \cite{McKeague&Qian2015} due to its heavy computational requirements. The parametric bootstrap analogue to ART is less computationally intensive than the ART, but still requires estimating the $p\times p$ covariance matrix $\Sigma$ and simulating from a $N(0,\hat{\Sigma})$ distribution. Due to computational constraints, we only run this parametric bootstrap analogue for $p\le 2\,000$ and not for $p=30\,000$. We also compare our method to a Bonferroni-corrected $t$-test.

All simulations are run using 1\,000 Monte Carlo simulations in \texttt{R} \citep{R2014}.

\begin{figure}[ht]
  \centering
  \includegraphics[width=\linewidth]{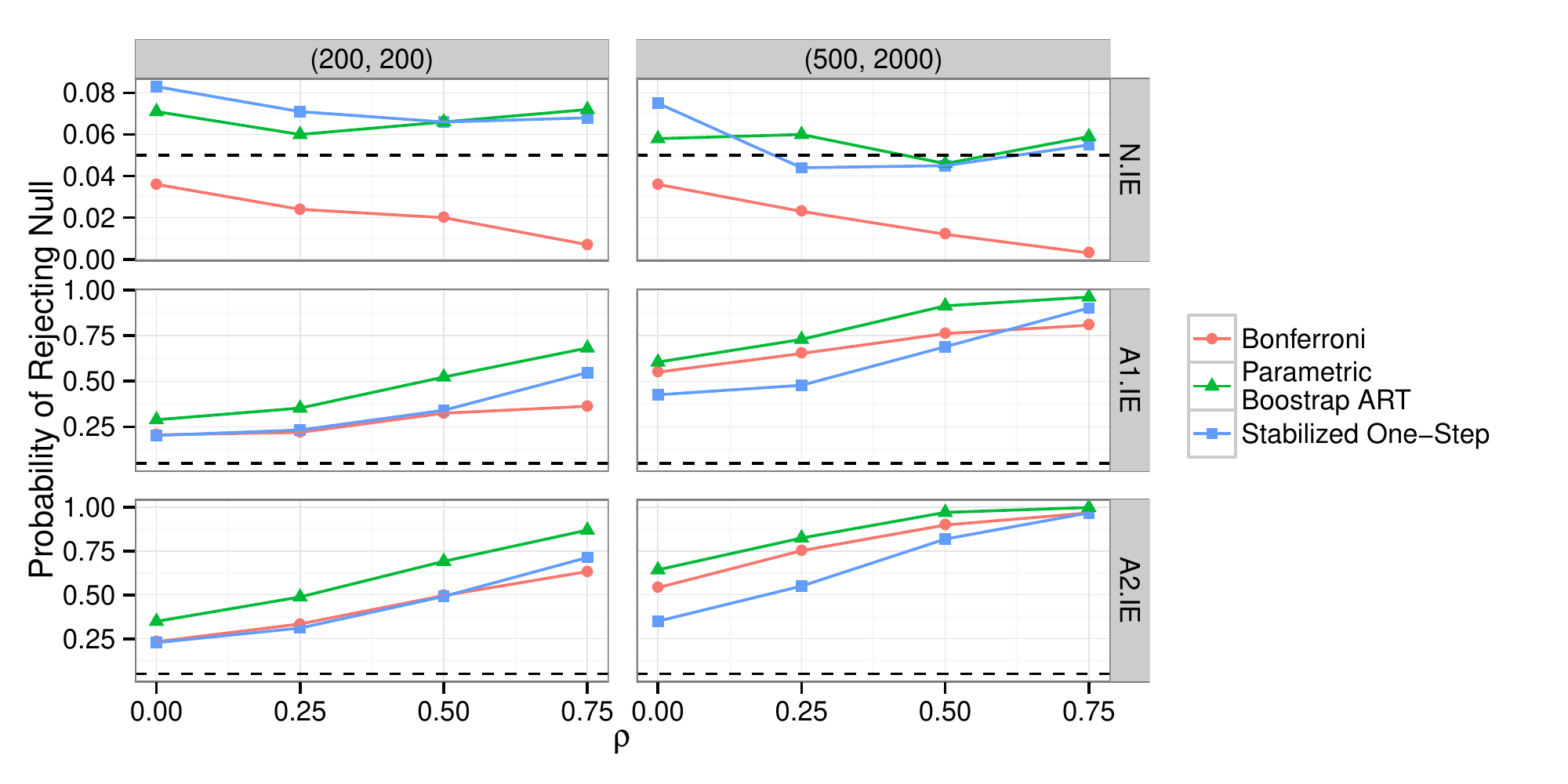}
  \caption{Power of the various testing procedures for $(n,p)$ equal to $(200,200)$ and $(500,2\,000)$ under homoscedastic errors. The parametric bootstrap analogue of ART performs the best in this setting.}
  \label{fig:smallIE}
\end{figure}

Figures \ref{fig:smallIE} displays the power of the three testing procedures for $(n,p)$ equal to $(200,200)$ and $(500,2\,000)$ for the homoscedastic data generating distributions N.IE, A1.IE, and A2.IE. The parametric bootstrap analogue of ART performs best in both of these settings. We can show (details omitted) that our method underperforms in this setting due to the second-order term representing the cost for estimating $d_0$ on subsets of the data of size $j\ll n$ early on in the procedure. While Theorem \ref{thm:cor2ndord} ensures that the estimate of $d_0$ will be asymptotically valid, there appears to be a noticeable price to pay at small sample sizes.

\begin{figure}[ht]
  \centering
  \includegraphics[width=\linewidth]{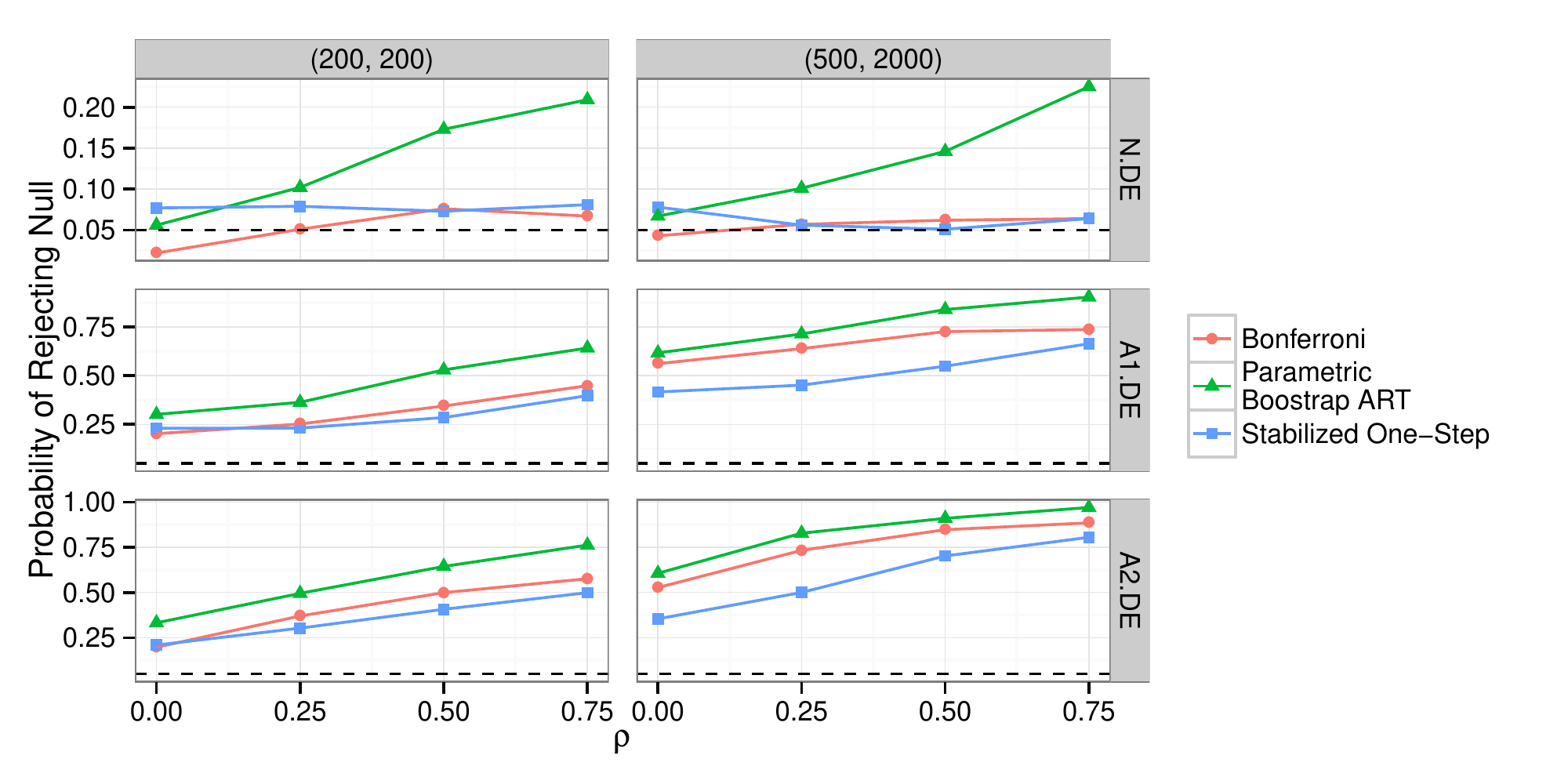}
  \caption{Power of the various testing procedures for $(n,p)$ equal to $(200,200)$ and $(500,2\,000)$ under heteroscedastic errors. The parametric bootstrap analogue of ART fails to control the type I error in this setting.}
  \label{fig:smallDE}
\end{figure}

Figures \ref{fig:smallDE} displays the power of the three testing procedures for $(n,p)$ equal to $(200,200)$ and $(500,2\,000)$ for the heteroscedastic data generating distributions. The parametric bootstrap analogue of ART fails to control the type I error in this setting. This is unsurprising given that this test was developed under a local linear model with independent errors. Both our method and Bonferroni adequately control type I error in this setting, especially at the larger sample size $n=500$, while we see that the Bonferroni procedure achieves slightly better power than our method for these data generating distributions.

\begin{figure}[ht]
  \centering
  \includegraphics[width=0.75\linewidth]{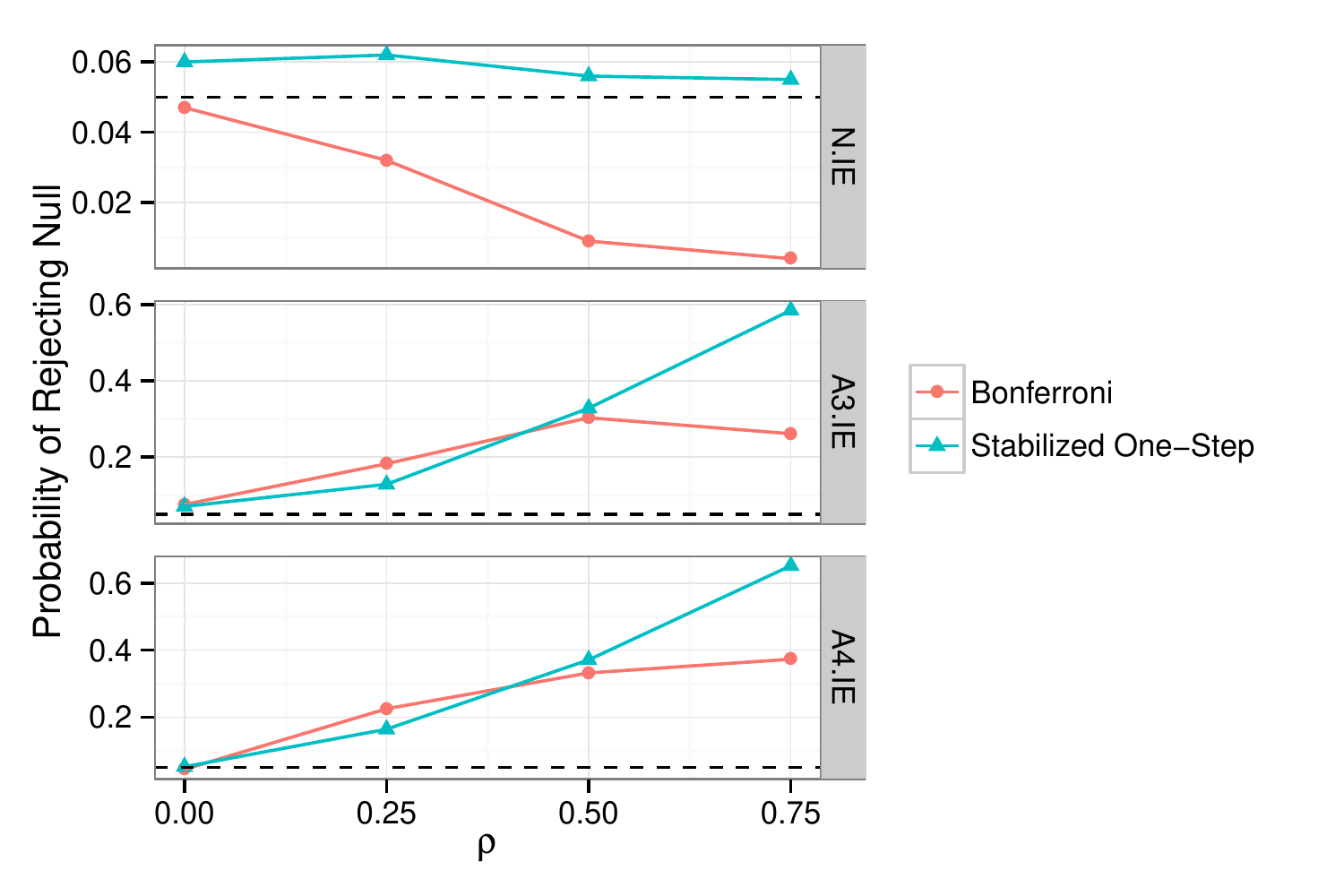}
  \caption{Power of the test from the stabilized one-step and from the Bonferroni-adjusted $t$-test for $(n,p)=(2\,000,30\,000)$ under homoscedastic errors. Unsurprisingly Bonferroni performs well when $\rho=0$. Our method outperforms the Bonferroni procedure when $\rho$ increases.}
  \label{fig:largeIE}
\end{figure}

Figure \ref{fig:largeIE} displays the power of our method and the Bonferroni procedure for $(n,p)$ equal to $(2\,000,30\,000)$. While (unsurprisingly) Bonferroni performs well when the correlation between the predictors in $X$ is low, our method outperforms the Bonferroni procedure when the correlation increases. We expect that, were we able to run the parametric bootstrap analogue of the ART at this sample size, it would outperform all other methods under consideration as it did at the smaller sample sizes. Nonetheless, this method quickly becomes computationally impractical when $p$ gets large, whereas our procedure and the Bonferroni procedure can still be implemented at these sample sizes. Furthermore, our method is robust to heteroscedastic errors and non-linear data generating distributions.

\section{Discussion}
We have presented a general method for estimating the (possibly non-unique) maximum of a family of parameter values indexed by $d\in\mathcal{D}_n$. Such an estimation problem is generally non-regular because minor fluctuations of the data generating distribution can change the subset of $\mathcal{D}_n$ for which the corresponding parameter is maximized. Our estimator takes the form of a sum of the terms of a martingale difference sequence, which quickly allows us to apply the relevant central limit theorem to study its asymptotics and develop Wald-type confidence intervals. The estimator adapts to the non-regularity of the problem, in the sense that we can give reasonable conditions under which it is regular and asymptotically linear when the maximizer is unique so that regularity is possible.

We have applied our approach to the example of \cite{McKeague&Qian2015} in which one wishes to learn about the maximal absolute correlation between a prespecified outcome and a predictor belonging to some set. The sample splitting that is built into our estimator has enabled us to analyze the estimator when the dimension $p$ of the predictor grows with sample size slowly enough so that $n^{-1/2}\log p\rightarrow 0$ as $n$ goes to infinity. While \citeauthor{McKeague&Qian2015} focus on testing the null hypothesis that this maximal absolute correlation is zero, we have established valid confidence intervals for this quantity. The lower bound of our confidence interval is particularly interesting because it is valid under minimal conditions. When $p$ is very large, one might expect that the null of no correlation between the outcome and any of the predictors is unlikely to be true. In these problems, having an estimate of the maximal absolute correlation, or at least a lower bound for this quantity, will likely still be interesting as a measure of the overall relationship between $X$ and $Y$. 

 We have also studied the behavior of this null hypothesis test under local alternatives, showing that our test is consistent when the maximal absolute correlation shrinks to zero slower than $n^{-1/2}(\log\max\{n,p\})^{1/2}$.
When the dimension of the predictor is fixed, the test of \citeauthor{McKeague&Qian2015} is consistent against alternatives shrinking to zero more slowly than $n^{-1/2}$ rather than $(\log n)^{1/2}n^{-1/2}$. We would not be surprised to find that this $(\log n)^{1/2}$ is unnecessary for $p$ fixed and can be removed using more refined proof techniques.


\citeauthor{McKeague&Qian2015} do not require that $Y$ and the coordinates of $X$ have range in $[-1,1]$. We have made this boundedness assumption out of convenience for our proofs and expect that we can replace the boundedness assumptions with appropriate moment assumptions without significantly changing the results. Our simulation results support this claim. The boundedness condition is not as restrictive as it may first seem, as unbounded $X$ and $Y$ can be rescaled to be to be bounded. Since the sharp null $H_0 : \Psi_n(P_0)=0$ is invariant to strictly monotonic transformations of $X$ and $Y$, our theoretical results yield a valid of $H_0$ test after applying, e.g., the sigmoid transformation to $X$ and $Y$.

We note that, in our simulations, the parametric bootstrap analogue of the ART achieves the highest power among competing methods in settings where we are able to run this procedure and the data generating distribution has homoscedastic errors. Nonetheless, this method is invalid under heteroscedastic errors, as we showed in our simulation. The theory for this method is also developed under a linear model, which will not exactly hold in practice. Furthermore, this procedure as currently described is computationally expensive and does not scale well to large data sets, especially when the dimension of the predictor $p$ is large. This difficulty occurs because the procedure requires the computation of a $p\times p$ covariance matrix. The earlier ART method presented in \cite{McKeague&Qian2015}, which achieves similar power to its parametric bootstrap analogue, has been shown to be even more computationally burdensome due to its use of a double bootstrap. Thus we believe our method represents an important contribution to the variable screening literature: it is computationally efficient, and has asymptotic theory supporting its power against local alternatives and increasing covariate dimension. None of the earlier works have given rigorous asymptotic theory when the dimension increases with sample size. Given our simulations, we also believe that developing rigorous asymptotic theory under increasing dimension, heteroscedastic errors, and nonlinear data generating distributions for the ART methods is an important area of future work.

The stabilized one-step estimator presented in this paper applies to many other situations not considered in this paper. In an earlier work, we showed that this estimator is useful for estimating the mean outcome under an optimal individualized treatment strategy \cite{Luedtke&vanderLaan2015b}, where the class $\mathcal{D}_n$ now indexes functions mapping from the covariate space to the set of possible treatment decisions. Thanks to the martingale structure of our estimator, the stabilized one-step estimator can be used to construct confidence intervals when the data is drawn sequentially so that the data generating distribution for observation $j$ can depend on that of the first $j-1$ observations. One interesting example along these lines is to obtain inference for the value of the optimal arm in a multi-armed bandit problem, even in the case where the optimal arm is non-unique and the reward distributions for the optimal arms have different variances. We look forward to seeing further applications of the general template for a stabilized one-step estimator that we have presented in this paper.

\bibliographystyle{Chicago}
 \bibliography{../../persrule}

\appendix
\setcounter{equation}{0}
\renewcommand{\theequation}{A.\arabic{equation}}
\setcounter{theorem}{0}
\renewcommand{\thetheorem}{A.\arabic{theorem}}
\renewcommand{\thecorollary}{A.\arabic{theorem}}
\renewcommand{\thelemma}{A.\arabic{theorem}}
\renewcommand{\theproposition}{A.\arabic{theorem}}
\renewcommand{\theconjecture}{A.\arabic{theorem}}

\section*{Appendix}



\section{Proofs and results for the \cite{McKeague&Qian2015} example}
\begin{lemma} \label{lem:rembd}
Fix $\tilde{\delta}>0$ and $d\in\mathcal{D}_n$. For any $P$ with $\min_{k\in\mathcal{K}_n} s_P^2(X_k)>\tilde{\delta}$ and $s_P^2(Y)>\tilde{\delta}$, (\ref{eq:rembd}) holds. 
\end{lemma}
\begin{proof}
Straightforward but tedious calculations show that
\begin{align}
\Rem^d(P)= m\Bigg(&\frac{1}{s_P(X_k)s_P(Y)}\left[s_P(X_k)s_P(Y)-s_0(X_k)s_0(Y)\right]\left[\Corr_P(X_k,Y)-\Corr_0(X_k,Y)\right] \nonumber \\
&+ \frac{\left(E_P[X_k]-E_{P_0}[X_k]\right)\left(E_P[Y]-E_{P_0}[Y]\right)}{s_P(X_k)s_P(Y)} \nonumber \\
&- \frac{\Corr_P(X_k,Y)}{2}\left[\frac{\left(E_P[X_k]-E_{P_0}[X_k]\right)^2}{s_P^2(X_k)} + \frac{\left(E_P[Y]-E_{P_0}[Y]\right)^2}{s_P^2(Y)}\right] \nonumber \\
&- \frac{\Corr_P(X_k,Y)}{2s_P^2(X_k)s_P^2(Y)}\left[s_P(X_k)s_0(Y)-s_0(X_k)s_P(Y)\right]^2\Bigg). \label{eq:exactrem}
\end{align}
The result follows by taking the absolute value of both sides, applying the triangle inquality, using that $ab\le (a^2+b^2)/2$ for any real $a,b$, $\Corr_P(X_k,Y)\le 1$, and the lower bound $\tilde{\delta}$ on the variances.
\end{proof}

We now establish high probability bounds on the difference between $s_{P_j}^2(X_k)$, $s_{P_j}^2(Y)$, $\Corr_{P_j}(X_k,Y)$, and $\hat{\sigma}_{nj}$ and their population counterparts, uniformly over $k\in\mathcal{K}_n$ and $j$. We will use $\lesssim$ to denote ``less than or equal to up to a universal multiplicative constant''. Let $\mathcal{F}_n$ denote the following class of functions mapping from $\mathcal{O}\equiv \mathbb{R}^\infty\times \mathbb{R}$ to the real line:
\begin{align}
\left\{(x,y)\mapsto x_k^r y^s : 0\le r,s,\le 4;\;r+s\le 4\;k\in\mathcal{K}_n\right\}. \label{eq:Fndef}
\end{align}
Note that $|\mathcal{F}_n|\lesssim p$. We will use this class to develop concentration results about our estimates the needed portions of the likelihood. This class is actually somewhat larger than is needed for most of our results, as in fact
\begin{align*}
\left\{(x,y)\mapsto x_k y : k\in\mathcal{K}_n\right\}&\cup \left\{(x,y)\mapsto x_k : k\in\mathcal{K}_n\right\}\cup\left\{(x,y)\mapsto y\right\} \\
&\cup \left\{(x,y)\mapsto x_k^2 : k\in\mathcal{K}_n\right\}\cup\left\{(x,y)\mapsto y^2\right\}
\end{align*}
suffices for concentrating our estimates of $\Corr_0(X_k,Y)$, $s_0(X_k)$, and $s_0(Y)$. Nonetheless, using this larger class $\mathcal{F}_n$ will allow us to prove results about the concentration of $\hat{\sigma}_{nj}^2$ about $\sigma_{nj}^2$, and just stating it as a single class is convenient for brevity.

For $f\in\mathcal{F}_n$ and $j\in\{1,\ldots,n\}$, define the empirical process as
\begin{align*}
\mathbb{G}_{nj}&\equiv \frac{1}{\sqrt{j}}\sum_{i=1}^j \left[f(O_i) - P_0 f\right] = \sqrt{j}(P_j-P_0)f,
\end{align*}
where we use $P_j$ denote the empirical distribution of $O_0,\ldots,O_{j-1}$ and $Pf\equiv E_P[f(O)]$ for any distribution $P$. Let $\Fnorm{\mathbb{G}_{nj}}\equiv \sup_{f\in\mathcal{F}_n}\left|\mathbb{G}_{nj}\right|$. By Theorem 2.14.1 in \cite{vanderVaartWellner1996} shows that
\begin{align}
E\Fnorm{\mathbb{G}_{nj}}&\lesssim \sqrt{\log\#\mathcal{F}_n}\lesssim \sqrt{\log p} \label{eq:expbd},
\end{align}
where the expectation is over the draws $O_1,\ldots,O_j$. We have used that our class is bounded by the constant $1$.

Let
\begin{align}
 K_{nj}&\equiv j^{-1/2}\sqrt{\log\max\{n,p\}}. \label{eq:Knj}
\end{align}
Define the events
\begin{align*}
\mathcal{A}_{nj}&\equiv \left\{\max_{f\in\mathcal{F}_n}\left|(P_j-P_0)f\right|\le CK_{nj}\right\}\textnormal{  for all }j=1,\ldots,n, \\
\mathcal{A}_n&\equiv \cap_{j=1}^n \mathcal{A}_{nj},
\end{align*}
where $C$ in the definition of $\mathcal{A}_{nj}$ is equal the smallest universal constant satisfying (\ref{eq:expbd}) plus $1$.

\begin{lemma} \label{lem:highprob}
For any sample size $n$, the event $\mathcal{A}_n$ occurs with probability at least $1-n/\max\{n^2,p\}\ge 1-1/n$.
\end{lemma}
\begin{proof}
We first upper bound the probability of the complement of $\mathcal{A}_{nj}$ for each $n,j$. Fix $n$ and $j\le n$. By the bounds on $X$ and $Y$, changing one $O_i$ in $(O_1,\ldots,O_j)$ to some other value in the support of $P_0$ can change $b$ by at most $1/\sqrt{j}$. Thus $(O_1,\ldots,O_j)\mapsto \Fnorm{\mathbb{G}_{nj}}$ satisfies the bounded differences property with bound $1/\sqrt{j}$, and we may apply McDiarmid's inequality \citep{McDiarmid1989} to show that, with probability at most $1-\exp(-2t^2)$, $\Fnorm{\mathbb{G}_{nj}}\le E\Fnorm{\mathbb{G}_{nj}} + t$. Choosing $t=\sqrt{\frac{\log\max\{n^2,p\}}{2}}$ and using (\ref{eq:expbd}) yields that, with probability at least $1-1/\max\{n^2,p\}$, the following inequality holds for all $j=1,\ldots,n$:
\begin{align*}
\Fnorm{\mathbb{G}_{nj}}&\le E\Fnorm{\mathbb{G}_{nj}} + \sqrt{\frac{\log\max\{n^2,p\}}{2}}\le C'\sqrt{\log p} + \sqrt{\log\max\{n,p\}} \\
&\le C\sqrt{\log\max\{n,p\}},
\end{align*}
where $C'$ denotes the universal constant in (\ref{eq:expbd}).

By DeMorgan's laws and a union bound, it follows that the event $\mathcal{A}_n\equiv\cap_j \mathcal{A}_{nj}$ occurs with probability at least $1-n/\max\{n^2,p\}\ge 1-1/n$.
\end{proof}
We have shown that $\mathcal{A}_n$ occurs with high probability. Now we show that our estimates of variances, covariances, and correlations perform well when $\mathcal{A}_n$ occurs.
\begin{lemma} \label{lem:concests}
Fix a sample size $n\ge 2$. The occurrence of $\mathcal{A}_n$ implies that, for all $j=2,\ldots,n$:
\begin{enumerate}[series=corr,noitemsep,label=\arabic*)]
  \item $\max_{k\in\mathcal{K}_n}\left|s_{P_j}(X_k)-s_0(X_k)\right|\lesssim \delta^{-1/2}K_{nj}$; \label{it:spnXs0X}
  \item $\max_{k\in\mathcal{K}_n}\left|s_{P_j}^2(X_k)-s_0^2(X_k)\right|\lesssim K_{nj}$; \label{it:spnXs0Xsq}
  \item $\left|s_{P_j}(Y)-s_0(Y)\right|\lesssim \delta^{-1/2}K_{nj}$; \label{it:spYs0Y}
  \item $\left|s_{P_j}^2(Y)-s_0^2(Y)\right|\lesssim K_{nj}$; \label{it:spYs0Ysq}
  \item $\max_{k\in\mathcal{K}_n}\left|\Corr_{P_j}(X_k,Y)-\Corr_{P_0}(X_k,Y)\right|\lesssim \delta^{-1}K_{nj}$, \label{it:corr}
\end{enumerate}
where we define $\Corr_{P_j}(X_k,Y)=0$ when either $s_{P_j}(X_k)$ or $s_{P_j}(Y)$ is equal to zero.
\end{lemma}
\begin{proof}
Suppose $\mathcal{A}_n$ holds and fix $k\in\mathcal{K}_n$. The triangle inequality and the bounds on $X_k$ yield that
\begin{align*}
\left|s_{P_j}^2(X_k) - s_0^2(X_k)\right|&= \left|\left(E_{P_j}[X_k^2]-E_{P_0}[X_k^2]\right)-\left(E_{P_j}[X_k] + E_{P_0}[X_k]\right)\left(E_{P_j}[X_k] - E_{P_0}[X_k]\right)\right| \nonumber \\
&\le \left|E_{P_j}[X_k^2]-E_{P_0}[X_k^2]\right| + 2\left|E_{P_j}[X_k] - E_{P_0}[X_k]\right| \lesssim K_{nj}.
\end{align*}
This gives \ref{it:spnXs0Xsq}. For \ref{it:spnXs0X}, note that
\begin{align*}
\left|s_{P_j}(X_k) - s_0(X_k)\right| = \left|\frac{s_{P_j}^2(X_k) - s_0^2(X_k)}{s_{P_j}(X_k) + s_0(X_k)}\right|\lesssim \delta^{-1/2} K_{nj}.
\end{align*}
The same argument yields \ref{it:spYs0Y} and \ref{it:spYs0Ysq}.

Again fix $k$. An application of the triangle inequality and the bounds on $X_k$ and $Y$ readily yield that $\left|\Cov_{P_j}(X_k,Y)-\Cov_{P_0}(X_k,Y)\right|\lesssim K_{nj}$. Furthermore,
\begin{align*}
\Corr_{P_j}&(X_k,Y)-\Corr_{P_0}(X_k,Y) \\
=&\, \frac{\Cov_{P_j}(X_k,Y)-\Cov_{P_0}(X_k,Y)}{s_0(X_k)s_0(Y)} - \frac{\Corr_{P_j}(X_k,Y)s_{P_j}(Y)}{s_0(X_k)s_0(Y)}\left[s_{P_j}(X_k)-s_0(X_k)\right] \\
&- \frac{\Corr_{P_j}(X_k,Y)}{s_0(Y)}\left[s_{P_j}(Y)-s_0(Y)\right].
\end{align*}

Taking the absolute value of both sides, applying the triangle inequality, and using the lower bounds on $s_0(X_k)$ and $s_0(Y)$ and the upper bound on $\Corr_{P_j}(X_k,Y)$ yields that $\left|\Corr_{P_j}(X_k,Y)-\Corr_{P_0}(X_k,Y)\right|\lesssim \delta^{-1}K_{nj}$. This holds for all $k$, so \ref{it:corr} holds.
\end{proof}

\begin{lemma} \label{lem:spbds}
Let $C$ be the smallest universal constant in \ref{it:spnXs0Xsq} of that Lemma \ref{lem:concests}, and let $n$ be any natural number satisfying $n\ge \lceil 4C^2\delta^{-2}\log\max\{n,p\}\rceil\equiv J(n,\delta)$. Under these conditions, the occurrence of $\mathcal{A}_n$ implies that, for all $j=J(n,\delta),\ldots,n$,
\begin{enumerate}[resume*=corr]
  \item $\min_{k\in\mathcal{K}_n} s_{P_j}^2(X_k)\ge \delta/2$ and $\min_{k\in\mathcal{K}_n} s_{P_j}^2(Y)\ge \delta/2$; \label{it:spnawayzero}
  \item $\max_{k\in\mathcal{K}_n}\frac{s_{P_0}^2(X_k)}{s_{P_j}^2(X_k)}\le 2$ and $\frac{s_{P_0}^2(Y)}{s_{P_j}^2(Y)}\le 2$; \label{it:sdratio}
  \item $\left|\hat{\sigma}_{nj}^2-\sigma^2_{nj}\right|\lesssim \delta^{-2}K_{nj}$; \label{it:sighatapproxsig}
  \item $\left|\hat{\sigma}_{nj}^2 - \Var_{P_0}\left(D_{nj}(P_0)(O)\right)\right|\lesssim \delta^{-2}K_{nj} + \left(\widehat{\Rem}_{nj}\right)^2$. \label{it:sighatapproxsig0}
\end{enumerate}
\end{lemma}
\begin{proof}
By Lemma \ref{lem:concests}, \ref{it:spnXs0Xsq} holds, and using that $j\ge J(n,\delta)$, we see that
\begin{align*}
s_{P_j}^2(X_k)&= s_{P_0}^2(X_k) + s_{P_j}^2(X_k)-s_{P_0}^2(X_k)\ge s_{P_0}^2(X_k) - \max_{k\in\mathcal{K}_n} \left|s_{P_j}^2(X_k)-s_{P_0}^2(X_k)\right|\ge \delta/2.
\end{align*}
The same argument works for $s_{P_j}^2(Y)$, so \ref{it:spnawayzero} holds. Furthermore,
\begin{align*}
\frac{s_{P_0}^2(X_k)}{s_{P_j}^2(X_k)}&\le \frac{s_{P_0}^2(X_k)}{s_{P_0}^2(X_k) - CK_{nj}}= 1 + \frac{CK_{nj}}{s_{P_0}^2(X_k) - CK_{nj}}\le 1+2C\delta^{-1}K_{nj}\le 2,
\end{align*}
where the final two inequalities hold by \ref{it:spnawayzero}. This proves the first part of \ref{it:sdratio}, and the bound on $s_{P_0}^2(Y)/s_{P_j}^2(Y)$ holds by the same argument. For the second result, note that
\begin{align*}
\left|\hat{\sigma}_{nj}^2 - \sigma_{nj}^2\right|&\le \left|(P_j-P_0)\hat{D}_{nj}^2\right| + \left|(P_j\hat{D}_{nj})^2 - (P_0\hat{D}_{nj})^2\right| \\
&\le \left|(P_j-P_0)\hat{D}_{nj}^2\right| + \left|(P_j + P_0)\hat{D}_{nj}\right|\left|(P_j-P_0)\hat{D}_{nj}\right|. \\
\intertext{Using \ref{it:spnawayzero}, the bounds on $X$ and $Y$, and the triangle inequality shows that}
&\lesssim \left|(P_j-P_0)\hat{D}_{nj}^2\right| + \delta^{-1}\left|(P_j-P_0)\hat{D}_{nj}\right|\lesssim \delta^{-2}\Fnorm{\mathbb{G}_{nj}},
\end{align*}
where we have used that $\mathcal{F}_n$ contains all polynomials of $X_k,Y$ of degree at most $4$. By the occurrence of $\mathcal{A}_n$, the final line is upper bounded by a constant times $\delta^{-2}K_{nj}$. This yields \ref{it:sighatapproxsig}.

For \ref{it:sighatapproxsig0}, we will bound $\left|\sigma_{nj}^2 - \Var_{P_0}\left(D_{nj}(P_0)(O)\right)\right|$ and then combine this with \ref{it:sighatapproxsig} using the triangle inequality. We have that
\begin{align*}
&\left|\sigma_{nj}^2 - \Var_{P_0}\left(D_{nj}(P_0)(O)\right)\right|\le \left|P_0\left[\hat{D}_{nj}^2-D_{nj}(P_0)^2\right]\right|+\left(P_0 \hat{D}_{nj}\right)^2.
\intertext{Now we use that $P_0 \hat{D}_{nj} = -\Corr_{P_j}(X_k,Y) + \Corr_{P_0}(X_k,Y) + \widehat{\Rem}_{nj}$ and $(a+b)^2\le 2(a^2+b^2)$ for any real $a,b$ to see that $\left(P_0 \hat{D}_{nj}\right)^2\lesssim \max_{k}\left(\Corr_{P_j}(X_k,Y) - \Corr_{P_0}(X_k,Y)\right)^2 + \left(\widehat{\Rem}_{nj}\right)^2$. By \ref{it:corr} from Lemma \ref{lem:concests} and the fact that $j\ge J(n,\delta)$, the maximum over $k\in\mathcal{K}_n$ is bounded above by a constant times $\delta^{-2}K_{nj}^2\lesssim \delta^{-2}K_{nj}$. Continuing with the above,}
&\lesssim \left|P_0\left(\left[\hat{D}_{nj}+D_{nj}(P_0)\right]\left[\hat{D}_{nj}-D_{nj}(P_0)\right]\right)\right|+\delta^{-2}K_{nj} + \left(\widehat{\Rem}_{nj}\right)^2 \\
&\lesssim \delta^{-1}P_0\left|\hat{D}_{nj}-D_{nj}(P_0)\right|+\delta^{-2}K_{nj} + \left(\widehat{\Rem}_{nj}\right)^2 \\
&\lesssim \delta^{-2}\Fnorm{\mathbb{G}_{nj}}+\delta^{-2}K_{nj} + \left(\widehat{\Rem}_{nj}\right)^2 \\
&\lesssim \delta^{-2}K_{nj} + \left(\widehat{\Rem}_{nj}\right)^2,
\end{align*}
where we used \ref{it:spnawayzero} for the second to last inequality.
\end{proof}

\begin{lemma} \label{lem:remnjbd}
Suppose the conditions of Lemma \ref{lem:spbds}. Under these conditions, the occurrence of $\mathcal{A}_n$ implies that, for all $j=J(n,\delta),\ldots,n$,
\begin{enumerate}[resume*=corr]
  \item $\left|\widehat{\Rem}_{nj}\right|\lesssim \delta^{-5/2} \left(K_{nj}\right)^2$. \label{it:remnjbd}
\end{enumerate}
\end{lemma}
\begin{proof}
By Lemma \ref{lem:spbds}, $\min_{k\in\mathcal{K}_n} s_{P_j}^2(X_k) \ge \delta/2$ and $\min_{k\in\mathcal{K}_n} s_{P_j}^2(Y)\ge \delta/2$. By Lemma \ref{lem:rembd}, this yields
\begin{align*}
\left|\widehat{\Rem}_{nj}\right|\lesssim \delta^{-1}\max_{k\in\mathcal{K}_n}\Bigg(&\left|s_{P_j}(X_k)s_{P_j}(Y)-s_0(X_k)s_0(Y)\right|\left|\Corr_{P_j}(X_k,Y)-\Corr_0(X_k,Y)\right| \\
&+ \left(E_{P_j}[X_k]-E_{P_0}[X_k]\right)^2 + \left(E_{P_j}[Y]-E_{P_0}[Y]\right)^2 \\
&+ \frac{s_0^2(Y)}{s_{P_j}^2(Y)}\left[s_{P_j}(X_k)-s_0(X_k)\right]^2 + \frac{s_0^2(X_k)}{s_{P_j}^2(X_k)}\left[s_{P_j}(Y)-s_0(Y)\right]^2\Bigg).
\end{align*}
By the bounds on $X$ and $Y$ and the triangle inequality, $\left|s_{P_j}(X_k)s_{P_j}(Y)-s_0(X_k)s_0(Y)\right|\le \left|s_{P_j}(X_k)-s_{P_0}(X_k)\right| + \left|s_{P_j}(Y)-s_{P_0}(Y)\right|$. Applying \ref{it:sdratio} from Lemma \ref{lem:spbds} and the results of Lemma \ref{lem:concests} to the above yields the result.
\end{proof}

\begin{lemma} \label{lem:eifvarlb}
Let $\gamma$ be as defined in (\ref{eq:gamma}). For a constant $C(\gamma,\delta)>0$ relying on $\gamma$ and $\delta$ only, the occurrence of $\mathcal{A}_n$ implies that, for all $j=\lceil C(\gamma,\delta)\log\max\{n,p\}\rceil,\ldots,n$,
\begin{enumerate}[resume*=corr]
  \item $\hat{\sigma}_{nj}^2\ge \gamma/2$.
\end{enumerate}
\end{lemma}
\begin{proof}[Sketch of proof]
Suppose $\mathcal{A}_n$. By \ref{it:sighatapproxsig0} and \ref{it:remnjbd}, for all $j\ge J(n,\delta)$
\begin{align*}
\left|\hat{\sigma}_{nj}^2 - \Var_{P_0}\left(D_{nj}(P_0)(O)\right)\right|&\lesssim \delta^{-2}K_{nj} + \delta^{-10}(K_{nj})^4.
\end{align*}
It is easy to confirm that, for a universal constant $C>0$, the above yields that the left-hand side is upper bounded by $\gamma/2$ for all $j\ge C\gamma^{-1/2}\delta^{-2}\max\left\{\delta^{-3},\gamma^{-3/2}\right\}\log\max\{n,p\}\equiv C(\gamma,\delta)\log\max\{n,p\}\ge J(n,\delta)$. An application of the triangle inequality gives the result.
\end{proof}

The remainder of the results in this section are asymptotic in nature. We omit the dependence on $\delta$ and $\gamma$ in these statements as these quantities are treated as fixed as sample size grows. Throughout we assume that
\begin{align}
&\frac{\log\max\{n,p\}}{\ell_n}\rightarrow 0, \label{eq:ellnlb1} \\
&\beta_n^2\log\frac{n}{\ell_n}\rightarrow 0, \label{eq:ellnlb2} \\
&\limsup_{n\rightarrow\infty}\frac{\ell_n}{n}<1. \label{eq:ellnub}
\end{align}
In view of (\ref{eq:ellnlb1}) and (\ref{eq:ellnlb2}), we see that, roughly, $\ell_n$ grows faster than $\log\max\{n,p\}$ if $\beta_n$ goes to zero faster than $1/\sqrt{\log n}$ and at least as fast as $n\exp(-o(\beta_n^{-2}))$ if $\beta_n$ goes to zero more slowly than $1/\sqrt{\log n}$. Given an $\epsilon>0$, one possible choice of $\ell_n$ that satisfies these properties is
\begin{align*}
\ell_n&= \max\left\{(\log\max\{n,p\})^{1+\epsilon},n\exp(-\beta_n^{-2+\epsilon})\right\}
\end{align*}

We have the following result.
\begin{lemma} \label{lem:nlgenough}
For all $n$ large enough, $\ell_n\ge J(n,\delta)$ and $\ell_n\ge C(\gamma,\delta)\log\max\{n,p\}$ as defined Lemmas \ref{lem:spbds} and \ref{lem:eifvarlb}, respectively.
\end{lemma}
\begin{proof}
This is an immediate consequence (\ref{eq:ellnlb1}) of the fact that $\delta$ and $\gamma$ are fixed as sample size grows.
\end{proof}

\begin{theorem}
\ref{it:C1}, \ref{it:C2}, and \ref{it:C3} hold.
\end{theorem}
\begin{proof}
\textbf{\ref{it:C1}:} By Lemma \ref{lem:nlgenough}, we can apply \ref{it:spnawayzero} from Lemma \ref{lem:spbds} and Lemma \ref{lem:eifvarlb} provided $n$ is large enough. In that case, $\frac{\hat{D}_{nj}}{\hat{\sigma}_{nj}}\lesssim \delta^{-1}\gamma^{-1/2}$ for all $j\ge\ell_n$ provided $\mathcal{A}_n$ holds. By Lemma \ref{lem:highprob}, this then occurs with probability at least $1-1/n$, and thus \ref{it:C1} holds.\\\\
\textbf{\ref{it:C2}:} If $\mathcal{A}_n$ holds, then Lemmas \ref{lem:eifvarlb} and \ref{lem:nlgenough} show that, for all $n$ large enough,
\begin{align*}
\frac{1}{n-\ell_n}\sum_{j=\ell_n}^{n-1}\left|\frac{\sigma_{nj}^2}{\hat{\sigma}_{nj}^2}-1\right|&\le \frac{2\gamma^{-1}}{n-\ell_n}\sum_{j=\ell_n}^{n-1}\left|\hat{\sigma}_{nj}^2-\sigma_{nj}^2\right| \lesssim \frac{\gamma^{-1}\delta^{-2}}{n-\ell_n}\log\max\{n,p\}\sum_{j=\ell_n}^{n-1}j^{-1}
\end{align*}
By \ref{it:sighatapproxsig} in Lemma \ref{lem:spbds} and the fact that $\sum_a^b j^{-1}\le \int_{a-1}^b j^{-1}dj$, the right-hand side is has an upper bound proportional to $\frac{\gamma^{-1}\delta^{-2}}{n-\ell_n}\log\max\{n,p\} \log n$. This bound is $o(1)$ by (\ref{eq:ellnub}) and the fact that $\beta_n\rightarrow 0$. The fact that $\mathcal{A}_n$ occurs with probability approaching $1$ (Lemma \ref{lem:highprob}) yields \ref{it:C2}.\\\\
\textbf{\ref{it:C3}:} Suppose that $n$ is large enough so that the results of Lemma \ref{lem:nlgenough} apply. Also suppose that $\mathcal{A}_n$ occurs. We have that
\begin{align*}
\frac{1}{n-\ell_n}\sum_{j=\ell_n}^{n-1} \frac{\widehat{\Rem}_{nj}}{\hat{\sigma}_{nj}}&\lesssim \frac{\gamma^{-1/2}\delta^{-5/2}}{n-\ell_n}\sum_{j=\ell_n}^{n-1} \left(K_{nj}\right)^2 \tag{Lemmas \ref{lem:remnjbd}, \ref{lem:eifvarlb}, and \ref{lem:nlgenough}} \\
&= \frac{\gamma^{-1/2}\delta^{-5/2}}{n-\ell_n}\log\max\{n,p\}\sum_{j=\ell_n}^{n-1} j^{-1} \tag{Eq. \ref{eq:Knj}} \\
&\lesssim \frac{\gamma^{-1/2}\delta^{-5/2}}{n-\ell_n}\log\max\{n,p\}\log\frac{n}{\ell_n} \tag{$\sum_a^b j^{-1}\le \int_{a-1}^b j^{-1}dj$} \\
&= o\left([n-\ell_n]^{-1/2}\right). \tag{Eqs. \ref{eq:ellnlb2} and \ref{eq:ellnub}}
\end{align*}
The fact that $\mathcal{A}_n$ occurs with probability approaching $1$ (Lemma \ref{lem:highprob}) yields \ref{it:C3}.
\end{proof}

Let $k_{n0}$ be a possibly non-unique $k$ maximizer of $\left|\Corr_{P_0}(X_k,Y)\right|$. For each $r>0$, let $\mathcal{K}_n^r\subseteq \mathcal{K}_n$ denote the set of all $k\in\mathcal{K}_n$ such that $\left|\Corr_{P_0}(X_{k_{n0}})\right| - \left|\Corr_{P_0}(X_k,Y)\right|\le r$.

The upcoming theorem uses the following conditions to establish the validity of a hypothesis test of no effect and of the upper bound of our confidence interval, respectively:
\begin{enumerate}[label=M\arabic*)]
  \item \label{it:margin} For some sequence $\{t_n\}$ with $t_n\rightarrow+\infty$, there exists a sequence of non-empty subsets $\Knstar\subseteq\mathcal{K}_n$ such that, for all $n$,
  \begin{align*}
  \inf_{k_1\in \Knstar}\left|\Corr_{P_0}(X_{k_1},Y)\right|\ge \sup_{k_2\in \mathcal{K}_n\backslash \Knstar}\left|\Corr_{P_0}(X_{k_2},Y)\right| + t_n n^{-1/4}\beta_n.
  \end{align*}
  If $\Knstar = \mathcal{K}_n$, then the supremum on the right-hand side is taken to be zero.
  \item \label{it:twosidedci} The conditions of \ref{it:margin} hold, and also
  \begin{align*}
  \textnormal{Diam}(\Knstar)&\equiv \sup_{k_1,k_2\in\Knstar}\left(\left|\Corr_{P_0}(X_{k_1},Y)\right|-\left|\Corr_{P_0}(X_{k_2},Y)\right|\right)=o(n^{-1/2}).
  \end{align*}
\end{enumerate}
The first of these conditions will be used to establish the consistency of a null hypothesis significance test. The second of these conditions is similar to margin conditions used in classification, and will be used to establish the validity of our confidence interval.

\begin{theorem} \label{thm:cor2ndord}
\begin{align}
\frac{1}{n-\ell_n}\sum_{j=\ell_n}^{n-1} \hat{\sigma}_{nj}^{-1}\left[\Psi^{d_{nj}}(P_0)-\Psi_n(P_0)\right]&= O_{P_0}\left(n^{-1/4}\beta_n\right). \label{eq:C4result1}
\end{align}
If also \ref{it:margin}, then the right-hand side of the above can be tightened to $O_{P_0}\left(\textnormal{Diam}(\Knstar)\wedge n^{-1/4}\beta_n\right) + o_{P_0}(n^{1/2})$. If also \ref{it:twosidedci}, then \ref{it:C4} holds.
\end{theorem}
\begin{proof}
Suppose that $\mathcal{A}_n$ holds and $n$ is large enough so that the results of Lemma \ref{lem:nlgenough} apply. For each $j\ge \ell_n$, let $k_{nj}$ represent the $k\in\mathcal{K}_n$ which maximizes $\left|\Corr_{P_j}(X_k,Y)\right|$. Let $m_0=\sgn [\Corr_{P_0}(X_{k_{n0}},Y)]$ and $m_{nj}=\sgn [\Corr_{P_j}(X_{k_{nj}},Y)]$. Then, for a universal constant $C>0$,
\begin{align}
0\ge& m_0\Corr_{P_j}(X_{k_{n0}},Y) - \left|\Corr_{P_j}(X_{k_{nj}},Y)\right| \nonumber \\
=&\, \left[m_0\Corr_{P_0}(X_{k_{n0}},Y) - m_{nj}\Corr_{P_0}(X_{k_{nj}},Y)\right] \nonumber \\
&+ m_0\left[\Corr_{P_j}(X_{k_{n0}},Y)-\Corr_{P_0}(X_{k_{n0}},Y)\right] - m_{nj}\left[\Corr_{P_j}(X_{k_{nj}},Y)-\Corr_{P_0}(X_{k_{nj}},Y)\right] \nonumber \\
\ge&\, \Psi_n(P_0) - \Psi^{d_{nj}}(P_0) - 2\max_{k\in\mathcal{K}_n}\left|\Corr_{P_j}(X_k,Y)-\Corr_{P_0}(X_k,Y)\right| \nonumber \\
\ge&\, \Psi_n(P_0) - \Psi^{d_{nj}}(P_0) - C\delta^{-1}K_{nj}, \label{eq:djd0comp}
\end{align}
where the final inequality holds by Lemma \ref{it:corr}. Using that $\sum_{j=\ell_n}^{n-1} j^{-1/2}\lesssim \sqrt{n}$ and (\ref{eq:ellnub}),
\begin{align*}
\frac{1}{n-\ell_n}\sum_{j=\ell_n}^{n-1} K_{nj}\lesssim \log\max\{n,p\}\frac{\sqrt{n}}{n-\ell_n}\lesssim n^{-1/4}\beta_n.
\end{align*}
By Lemma \ref{lem:eifvarlb}, this then implies that the left-hand side of (\ref{eq:C4result1}) is upper bounded by an $O\left(\gamma^{-1/2}n^{-1/4}\beta_n\right)$ term under $\mathcal{A}_n$, and so Lemma \ref{lem:highprob} yields (\ref{eq:C4result1}).

For the second result, suppose that \ref{it:margin} holds. Observe that, for all $j> Cnt_n^{-1}$ for $C$ as defined in (\ref{eq:djd0comp}), $\Psi_n(P_0) - \Psi^{d_{nj}}(P_0)< t_n n^{-1/4}\beta_n$. Furthermore, $\Psi_n(P_0)-\left|\Corr_{P_0}(X_{k_{nj}})\right|\le \Psi_n(P_0) - \Psi^{d_{nj}}(P_0)$. Thus $k_{nj}\in \Knstar$ as defined in \ref{it:margin}. Furthermore, $m_{nj}$ must equal $\sgn[\Corr_{P_0}(X_{k_{nj}},Y)]$, since otherwise
\begin{align*}
\Psi_n(P_0) - \Psi^{d_{nj}}(P_0)&\ge \left|\Corr_{P_0}(X_{k_{nj}},Y)\right| - \Psi^{d_{nj}}(P_0)\ge 2\inf_{k\in\Knstar}\left|\Corr_{P_0}(X_{k_{nj}},Y)\right|\ge 2t_n n^{-1/4}\beta_n,
\end{align*}
contradicting the fact that $\Psi_n(P_0) - \Psi^{d_{nj}}(P_0)\le t_n n^{-1/4}\beta_n$ per (\ref{eq:djd0comp}). Because $k\in\Knstar$, we see that $\Psi^{d_{nj}}(P_0)\ge\inf_{k\in\Knstar}\left|\Corr_{P_0}\left(X_k,Y\right)\right|$. Hence,
\begin{align}
\frac{1}{n-\ell_n}\sum_{j=\max\{\ell_n,\lceil Cnt_n^{-1}\rceil\}}^{n-1} \left[\Psi^{d_{nj}}(P_0)-\Psi_n(P_0)\right]&\ge -\textnormal{Diam}(\Knstar). \label{eq:classifright}
\end{align}
Further, if $\lceil Cnt_n^{-1}\rceil\ge \ell_n$, (\ref{eq:djd0comp}) yields
\begin{align*}
\sum_{j=\ell_n}^{\lceil Cnt_n^{-1}\rceil} \left[\Psi^{d_{nj}}(P_0)-\Psi_n(P_0)\right]\ge -C\sum_{j=\ell_n}^{\lceil Cnt_n^{-1}\rceil} K_{nj}\ge -C\log\max\{n,p\}\int_{\ell_n-1}^{\lceil Cnt_n^{-1}\rceil} j^{-1/2}dj
\end{align*}
It follows that the left-hand side above is greater than or equal to a positive universal constant times $n^{1/2}t_n^{-1/2}$. Dividing the left by $n-\ell_n$ and applying (\ref{eq:ellnub}) yields that this same result holds with an upper bound on the order of $n^{-1/2}t_n^{-1/2}$. Combining this with (\ref{eq:classifright}) shows that
\begin{align*}
\frac{1}{n-\ell_n}\sum_{j=\ell_n}^{n-1} \left[\Psi^{d_{nj}}(P_0)-\Psi_n(P_0)\right]&\ge -\textnormal{Diam}(\Knstar) + O(n^{-1/2}t_n^{-1/2}).
\end{align*}
Using that $t_n^{-1/2}\rightarrow 0$, $n^{-1/2}t_n^{-1/2}=o(n^{-1/2})$. When proving the first result (\ref{eq:C4result1}) we also showed that the left-hand side is upper-bounded by a positive constant times $-\delta^{-1}n^{-1/4}\beta_n$. Combining with Lemma \ref{lem:eifvarlb} and using that $\mathcal{A}_n$ holds with probability approaching $1$ (Lemma \ref{lem:highprob}) shows that the left-hand side of (\ref{eq:C4result1}) is $O_{P_0}\left(\textnormal{Diam}(\Knstar)\wedge n^{-1/4}\beta_n\right) + o_{P_0}(n^{-1/2})$. If \ref{it:twosidedci} holds, then this expression is $o_{P_0}(n^{-1/2})$, and so \ref{it:C4} holds.
\end{proof}

\section{Computationally efficient implementation of our estimator for the \cite{McKeague&Qian2015} example} \label{app:compcomplex}
In this section, we describe how to implement the estimator in $O(np)$ time. We show that this can be accomplished using $O(p)$ storage when the observations $O_1,\ldots,O_n$ arrive in a stream.


Fix $n$ so that the set $\mathcal{K}_n$ of predictor indices is also fixed. For each $j$, let $P_j$ denote the empirical distribution of the first $j$ observations. Recall the definition of the class $\mathcal{F}_n$ from (\ref{eq:Fndef}), and note that $\mathcal{F}_n$ contains $O(p)$ functions. It is easy to see that, at $j=2$, we can compute $P_j f\equiv E_{P_j}[f(O)]$ for each $f\in\mathcal{F}_n$ using $O(p)$ time and storage. Furthermore, for $j>3$ the fact that $P_jf = f(O_j) + \frac{j-1}{j}P_{j-1}f$ shows that we can compute and save $P_jf$ in $O(p)$ time and storage if we know $O_j$ and $P_{j-1}f$. To attain this storage complexity, we remove $P_{j-2}f$, $f\in\mathcal{F}_n$, from memory for each $j\ge 4$ so that $P_2 f,\ldots,P_{j-2}f$ are not stored in memory.

We now have an algorithm that, at observation $j$, starts with $O_j$ and $P_{j-1} f$, $f\in\mathcal{F}_n$, stored in memory and, after running the steps described in the preceding paragraph, also has $P_j f$, $f\in\mathcal{F}_n$ stored in memory. Given $P_j f$, $f\in\mathcal{F}_n$, one can compute and save $\Cov_{P_j}(X_k,Y)=E_{P_j}[X_k Y] - E_{P_j}[X_k]E_{P_j}[Y]$, $k\in\mathcal{K}_n$, and $s_{P_j}^2(Z)=E_{P_j}[Z^2]-E_{P_j}[Z]^2$, $Z$ equal to $Y$ or $X_k$, $k\in\mathcal{K}_n$, in $O(p)$ time and storage. We can now compute and save $\Corr_{P_j}(X_k,Y)=\frac{\Cov_{P_j}(X_k,Y)}{s_{P_j}(X_k)s_{P_j}(Y)}$, $k\in\mathcal{K}_n$, in $O(p)$ time and storage. If the predictors or outcome are large and their variance small, the described online computation of the sample variance may lead to numerical difficulties. See \cite{Welford1962} for a better estimate of the variance in this setting.

Let $H_j$ denote the collection of (i) the integer $j$, (ii) $P_j f$, $f\in\mathcal{F}_n$, (iii) $s_{P_j}^2(Y)$, and (iv) $\Cov(X_k,Y)$, $s_{P_j}^2(X_k)$ and $\Corr_{P_j}(X_k,Y)$, $k\in\mathcal{K}_n$. For $j\ge 2$, let \textproc{UpdateH} be a function which takes as input $(O_{j+1},H_j)$ and outputs $H_{j+1}$. We have shown that $\textproc{UpdateH}(O_{j+1},H_j)$ can run in $O(p)$ time for any $j\ge 2$. We call a separate function \textproc{InitializeH} on $(O_1,O_2)$ to obtain the initial value $H_2$. This function runs in $O(p)$ time and storage.

Let the function \textproc{Maximizer} be a function that takes as input $H_j$ and returns the $d_j=(k_j,m_j)$ which maximizes $m\Corr_{P_j}(X_k,Y)$ in $d\equiv (k,m)\in\mathcal{K}_n$, thereby allowing us to compute $\hat{\sigma}_{nj}=P_j D^{d_j}(P_j)^2$. Finding $d_j$ involves finding the maximum of $|\mathcal{D}_n|=2p$ numbers, and therefore can be accomplished in $O(p)$ time.

The function \textproc{CalcD} takes as input $H_j$, $O_{j+1}$, and $d_j$ and calculates $D^{d_j}(P_j)(O_{j+1})$. It is easy to see that this can be accomplished in $O(1)$ time and $O(p)$ storage.

For ease of notation in the proceeding paragraph and equation we omit the dependence of $d_j=(k_j,m_j)$ on $j$ in the notation. Since $D^d(P_j)$ is a gradient for $\Psi^d$ at $P_j$ and gradients are mean zero, $P_j D^d(P_j)=0$. For any $d\in\mathcal{D}_n$, tedious but trivial calculations show that
\begin{align*}
P_j D^d(P_j)^2=&\, \left[\frac{2 + \Corr_{P_j}(X_k,Y)^2}{2s_{P_j}^{2}(X_k)s_{P_j}^{2}(Y)}\right] \sum_{r=0}^2 \sum_{s=0}^2 (-1)^{r+s}\binom{2}{r}\binom{2}{s} E_{P_j}[X_k^r Y^s] E_{P_j}[X_k]^{2-r} E_{P_j}[Y]^{2-s} \\
&+ \frac{\Corr_{P_j}(X_k,Y)^2}{4}\sum_{r=0}^4 (-1)^r\binom{4}{r} \left[\frac{E_{P_j}[X_k^r]E_{P_j}[X_k]^{4-r}}{s_{P_j}^{4}(X_k)}+\frac{E_{P_j}[Y^r]E_{P_j}[Y]^{4-r}}{s_{P_j}^{4}(Y)}\right] \\
&- \frac{\Corr_{P_j}(X_k,Y)}{s_{P_j}(X_k)^3 s_{P_j}(Y)}\sum_{r=0}^3 \sum_{s=0}^1 (-1)^{r+s} \binom{3}{r}E_{P_j}[X_k^r Y^s] E_{P_j}[X_k]^{3-r} E_{P_j}[Y]^{1-s} \\
&- \frac{\Corr_{P_j}(X_k,Y)}{s_{P_j}(X_k) s_{P_j}(Y)^3}\sum_{r=0}^1 \sum_{s=0}^3 (-1)^{r+s} \binom{3}{s}E_{P_j}[X_k^r Y^s] E_{P_j}[X_k]^{1-r} E_{P_j}[Y]^{3-s}.
\end{align*}
Observe that all expectations on the right-hand side above are expectations over some $f\in\mathcal{F}_n$ applied to the observed data structure. It follows that the above can be computed in $O(1)$ time using a subset of the $O(p)$ expectation, standard deviation, and correlation estimates stored in $H_j$. Let \textproc{CalcSigHat} denote the function which takes as input $H_j$ and $d_j$ and outputs $\hat{\sigma}_{nj}$. We have shown that \textproc{CalcSigHat}$(H_j,d_j)$ runs in $O(1)$ time.

The pseudocode in \textproc{EstPsi} describes our estimator, with most of the work done in the recursion step described in the function \textproc{Recursion}. Because each call of \textproc{Recursion} runs in $O(p)$ time, the $n-\ell_n=O(n)$ step for loop in \textproc{EstPsi} requires time $O(np)$ time. The storage requirement of each call of \textproc{Recursion} is $O(p)$. Because the code in the for loop in \textproc{EstPsi} deletes the output from the previous recursion step, the total storage requirement of \textproc{EstPsi} is $O(p)$.

\begin{algorithm}[ht]
\caption{Recursion Step for Estimating $\Psi(P_0)$}\label{alg:recursion}
\begin{algorithmic}[2]
\Function{Recursion}{$O_{j+1}$, $\psi_j$, $H_j$, $\bar{\sigma}_j$, $\ell_n$}
\If{$j<\ell_n$} $\psi_{j+1} = 0$ and $\bar{\sigma}_{j+1} = 0$
\Else
  \State $d_j = \textproc{Maximizer}(H_j)$
  \State $\hat{\sigma}_{nj} = \textproc{CalcSigHat}(H_j,d_j)$
  \State $D^{d_j}(P_j)(O_{j+1}) = \textproc{CalcD}(H_j,O_{j+1})$
  \State $\psi_{j+1} = \frac{\psi_j}{\bar{\sigma}_j} + \frac{\Corr_{P_j}(X_{d_j},Y) + D^{d_j}(P_j)(O_{j+1})}{\hat{\sigma}_{nj}}$\Comment{By convention, $0/0=0$.}
  \State $\bar{\sigma}_{j+1} = \frac{1}{j+1}\left[j\bar{\sigma}_j + \hat{\sigma}_{nj}\right]$
\EndIf
\State $H_{j+1} = \textproc{UpdateH}(O_{j+1},H_j)$
\State \Return $\left(\psi_{j+1},\bar{\sigma}_{j+1},H_{j+1}\right)$
\EndFunction
\end{algorithmic}
\end{algorithm}

\begin{algorithm}[ht]
\caption{Estimate $\Psi(P_0)$ Using Sample of Size $n$}\label{alg:SL}
\begin{algorithmic}[2]
\Function{EstPsi}{$n$, $\ell_n$}
\State Read $O_1,O_2$ from data stream
\State Base case: $\psi_2 = 0$, $\bar{\sigma}_2 = 0$, and $H_2 = \textproc{InitializeH}(O_1,O_2)$
\For{$j=2,\ldots,n-1$}
  \State Read $O_{j+1}$ from data stream
  \State $\left(\psi_{j+1},\bar{\sigma}_{j+1},H_{j+1}\right) = \textproc{Recursion}\left(O_{j+1},\psi_j,H_j,\bar{\sigma}_j,\ell_n\right)$
  \State Remove $\left(O_{j+1},\psi_j,H_j,\bar{\sigma}_j\right)$ from memory
\EndFor
\State \Return Point estimate $\psi_n$ and confidence interval $\left[\psi_n\pm 1.96\bar{\sigma}_n/\sqrt{n}\right]$
\EndFunction
\end{algorithmic}
\end{algorithm}

\end{document}